\documentclass[a4paper,10pt,DIV=calc]{scrartcl}

\usepackage{cancel}
\usepackage{tabularx}
\usepackage{subfig}
\usepackage{a4}
\usepackage{amsmath,amssymb}
\usepackage{bbm}
\usepackage{amsthm}
\usepackage{dsfont}
\usepackage[onehalfspacing]{setspace} 
\usepackage{appendix}
\usepackage{graphicx}
\usepackage{enumerate}
\usepackage{hyphenat}
\usepackage{cite}
\usepackage{colonequals}
\usepackage{relsize}
\usepackage{caption}

\usepackage{MnSymbol}

\usepackage{enumerate}
\usepackage{relsize}

\usepackage[usenames,dvipsnames]{xcolor}
\usepackage[all]{xypic}

\usepackage{accents}
\newlength{\dhatheight}

 \theoremstyle{plain} \newtheorem{Theorem}{Theorem}
\theoremstyle{plain} \newtheorem{Corollary}{Corollary}
\theoremstyle{plain} \newtheorem{Lemma}{Lemma}
\theoremstyle{plain} \newtheorem{KeyLemma}[Lemma]{Key Lemma}
\theoremstyle{plain} \newtheorem{Proposition}{Proposition}
\theoremstyle{definition} \newtheorem{Def}{Definition}
\theoremstyle{plain} \newtheorem{LemDef}[Def]{Lemma and Definition}
\theoremstyle{definition}\newtheorem{Remark}{Remark}
\theoremstyle{definition}
\theoremstyle{definition} \newtheorem{Example}{Example}
\newcommand{\dist}{\operatorname{dist}}
\newcommand{\dm}{d}
\newcommand{\odm}{\bar{\dm}}

\newcommand{\id}{{\mathbb{I}}}

\DeclareMathOperator{\Tr}{Tr} 
\DeclareMathOperator*{\argmin}{arg\,min}
\newcommand{\ro}{{\mathbb{R}}}

\newcommand{\ra}{\rightarrow}

\newcommand{\no}{\nonumber}

\newcommand{\spann}{\operatorname{span}}

\newcommand{\for}{\mbox{ for }}

\newcommand{\Lra}{\Leftrightarrow}

\newcommand{\SO}{{\mathrm{SO(3)}}}
\newcommand{\Sl}{{\mathrm{SL(3,\mathbb{Z})}}}

\newcommand{\Slk}{{\mathrm{SL^k(3,\mathbb{Z})}}}
\newcommand{\SL}[1]{{\mathrm{SL^{#1}(3,\mathbb{Z})}}}

\newcommand{\zp}{\mathbb{Z}^{3\times 3}_+}
\newcommand{\slz}{\mathrm{SL(3,\mathbb{Z})}}

\newcommand{\glp}{\mathrm{GL^+(3,\mathbb{R})}}
\newcommand{\col}[1]{\mathrm{col}[\mathnormal{#1}]}

\newcommand{\rr}{\mathbb{R}^{3\times 3}}
\newcommand{\R}{{\mathbb{R}^3}}

\newcommand{\s}{\sigma}

\newcommand{\diag}{\operatorname{diag}}

\newcommand{\al}{\alpha}

\mathchardef\mdash="2D

\newcommand{\raf}{\ra \infty}

\newcommand{\ssp}[1]{\langle #1 \rangle}
\newcommand{\tinf}[1]{\| #1 \|_{2,\infty}}
\newcommand{\PP}{\mathcal{P}^{24}}

\newcommand{\aand}{\mbox{ and }}

\newcommand{\ml}{\mathcal{L}}
\newcommand{\ac}{{\mbox{\tiny{AC}}}}

\newcommand{\dop}{\ensuremath{\operatorname{d}}}
\newcommand{\la}{\lambda}
\newcommand{\etc}{{\normalfont{[...]}\ }}
\newcommand{\quo}[1]{\emph{``#1''}}
\newcommand{\icon}{$\ast$}
\newcommand{\T}{^\mathrm{T}\!} 

\def\signed #1{{\leavevmode\unskip\nobreak\hfil\penalty50\hskip2em
  \hbox{}\nobreak\hfil #1%
  \parfillskip=0pt \finalhyphendemerits=0 \endgraf}}

\newsavebox\mybox
\newenvironment{aquote}[1]
  {\savebox\mybox{#1}\begin{quote}}
  {\signed{\usebox\mybox}\end{quote}}
  
  \usepackage{authblk}

\usepackage{fancyhdr}
\usepackage{hyperref}
\makeatletter
\def\@maketitle{%
  \newpage
  \null
  \vskip 2em%
  \begin{center}%
  \let \footnote \thanks
    {\Large\bf  \@title \par}%
    \vskip 1.5em%
    {\normalsize
      \lineskip .5em%
      \begin{tabular}[t]{c}%
        \@author
      \end{tabular}\par}%
    \vskip 1em%
    {\normalsize \@date}%
  \end{center}%
  \par
  \vskip 1.5em}
\makeatother
\title{Optimality of General Lattice Transformations with Applications to the Bain Strain in Steel}

\author{Konstantinos Koumatos%
  \thanks{\texttt{konstantinos.koumatos@gssi.infn.it}}}
\affil{\small\textit{Gran Sasso Science Institute, }\\ \textit{Viale Fransesco Crispi 7,} \\ \textit{67100 L'Aquila, Italy}}

\author{Anton Muehlemann%
  \thanks{\texttt{muehlemann@maths.ox.ac.uk}}}
\affil{\small\textit{Mathematical Institute, University of Oxford,}\\\textit{ Andrew Wiles Building, Radcliffe Observatory Quarter, Woodstock Road,} \\ \textit{Oxford OX2 6GG, United Kingdom}}

\date{Dated: \today}

 \hypersetup{
  pdfauthor = {Konstantinos Koumatos (Gran Sasso Science Institute) and Anton Muehlemann (University of Oxford)},
  pdftitle = {Optimality of general lattice transformations with applications to the Bain strain in steel},
  pdfsubject = {MSC (2010): 74N05, 74N10},
  pdfkeywords = {lattice transformation, Bain strain, stretch tensor, optimal transformation, least atomic movement, steel, least atomic motion, Terephthalic acid, fcc-to-bcc, fcc-to-bct, Bravais lattices, austenite, martensite, low carbon steel}
}

\begin{document}

\maketitle
\begin{abstract}
 This article provides a rigorous proof of a conjecture by E.C. Bain in 1924 on the optimality of the so-called ``Bain strain'' based on a criterion of least atomic movement. A general framework that explores several such optimality criteria is introduced and employed to show the existence of optimal transformations between any two Bravais lattices. A precise algorithm and a \hyperref[Matlab]{GUI} to determine this optimal transformation is provided. Apart from the Bain conjecture concerning the transformation from face-centred cubic to body-centred cubic, applications include the face-centred cubic to body-centred tetragonal transition as well as the transformation between two triclinic phases of Terephthalic Acid.
 \vspace{4pt}

\noindent\textsc{MSC (2010): 74N05, 74N10} 

\noindent\textsc{Keywords:} lattice transformation, least atomic movement, Bravais lattices, Bain strain in steel, fcc-to-bcc, Terephthalic Acid

\vspace{4pt}
\end{abstract}
\newpage
\tableofcontents
\section{Introduction}
 In his seminal article on ``The nature of martensite'' Bain \cite{Bain} proposed a mechanism that transforms the face-centred cubic (fcc) lattice to the body-centred cubic (bcc) lattice, a phase transformation most importantly manifested in low carbon steels. He writes:
\begin{aquote}{Bain, 1924}
 \quo{It is reasonable, also, that the atoms themselves will rearrange \etc by a method that will require least temporary motion. \etc A mode of atomic shift requiring minimum motion was conceived by the author \etc} 
\end{aquote}
The key observation that led to his famous correspondence was that \quo{If one regards the centers of faces as corners of a new unit, a body-centered structure is already at hand; however, it is tetragonal instead of cubic}. He remarks that this is not surprising \quo{as it is the only easy method of constructing a bcc atomic structure from the fcc atomic structure}.

Even though now widely accepted, his mechanism, which he illustrated with a model made of cork balls and needles (see Fig.~\ref{BainOrg}), was not without criticism from his contemporaries. In their fundamental paper Kurdjumov \& Sachs \cite{KS} wrote [free translation from German] that ``nothing certain about the mechanism of the martensite transformation is known. Bain imagines that a tetragonal unit cell within the fcc lattice transforms into a bcc unit cell through compression along one direction and expansion along the two other. However a proof of this hypothesis is still missing''.\footnote{\quo{\"Uber den Mechanismus dieser ,,Martensitumwandlung" ist bisher
nichts Sicheres bekannt. Bain stellt sich vor, da{\ss} eine tetra- gonalk\"orperzentrierte
Elementarzelle des Austenits durch Schrumpfung in der
einen Richtung und Ausdehnung in den beiden anderen in die kubischraumzentrierte
des $\al$-Eisens \"ubergeht. Eine Best\"atigung f\"ur diese Anschauung
konnte bisher nicht erbracht werden.}} Interestingly, without being aware of it, the authors implicitly used the Bain mechanism in their derivation of the Kurdjumov \& Sachs orientation relationships (see \cite{KSpapersmall} for details).
  \begin{figure}[h]
  \centering
  {\includegraphics[height=5cm]{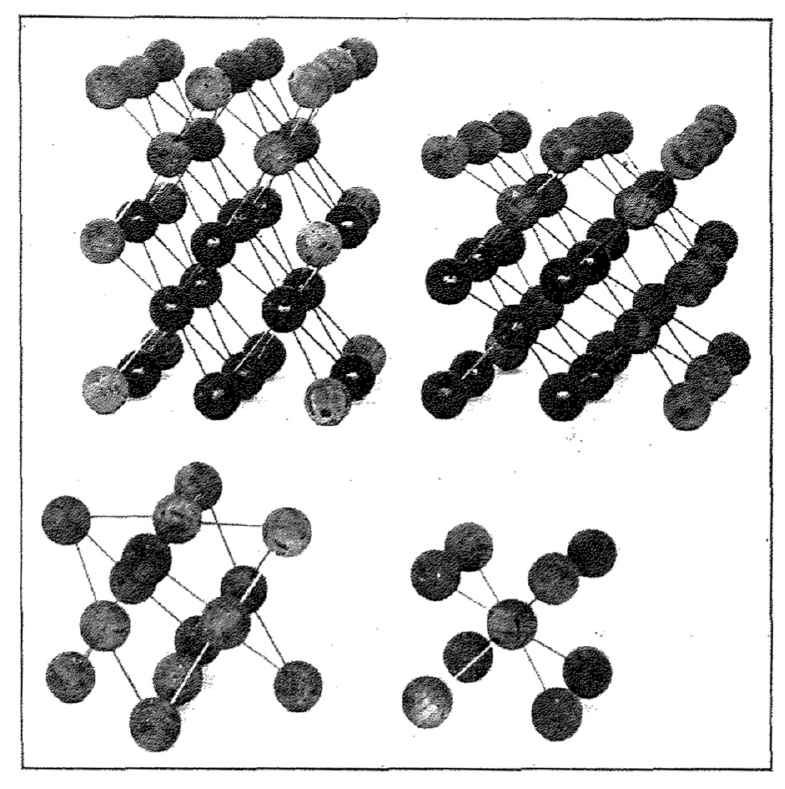}} 
  \caption{From E.C. Bain: the small models show the fcc and bcc unit cells; the large models represent 35 atoms in an fcc and bcc arrangement respectively.} \label{BainOrg}
\end{figure}

\noindent In subsequent years, the determination of the transformation mechanism remained of great interest. In their paper on ``Atomic Displacements in the Austenite-Martensite Transformation`` \cite{Jaswon} Jaswon and Wheeler again acknowledged that 
\begin{aquote}{Jaswon and Wheeler, 1948}
 \quo{Off all the possible distortions of a primitive unit cell of the face-centred cubic structure, which could generate a body-centred cubic structure of the given relative orientation, the one which actually occurs is the smallest} 
\end{aquote}
By combining it with experimental observations of the orientation relationships they devised an algorithm to derive the strain tensor. However their approach is only applicable to cases where the orientation relationship is known a-priori.

With the years passing and a number of supporting experimental results (for a discussion see e.g. \cite{BowlesWayman}) the Bain mechanism rose from a conjecture to a widely accepted fact. Nevertheless, almost a century after Bain first announced his correspondence a rigorous proof based on the assumption of minimal atom movement has been missing.
Of course, the transformation from fcc-to-bcc is not the only instance where the determination of the transformation strain is of interest. The overall question remains the same: Which transformation strain(s), out of all the possible deformations mapping the lattice of the parent phase to the lattice of the product phase, require(s) the least atomic movement?

To provide a definite answer to this question one first needs to quantify the notion of least atomic movement in such a way that it does not require additional input from experiments. Then one needs to establish a framework that singles out the optimal transformation among the infinite number of possible lattice transformation strains. One way to appropriately quantify least atomic movement is the criterion of smallest principal strains as suggested by Lomer in \cite{Lomer}. In his paper, Lomer compared $1600$ different lattice correspondences for the $\beta$ to $\al$ phase transition in Uranium and concluded that only one of them involved strains of less than $10\%$. More recently, in \cite{Sherry} an algorithm is proposed to determine the transformation strain based on a similar minimality criterion (see Remark~\ref{RemarkDistances}) that also allows for the consideration of different sublattices. The present paper considers a criterion of least atomic movement in terms of a family of different strain 
measures and, for each such strain measure, rigorously proves the existence of an optimal lattice transformation between \emph{any} two Bravais lattices.\footnote{In particular, no assumptions are made on the type of lattice points (e.g. atoms, molecules) or on the relation between the point groups of the two lattices.} As a main application, it is shown that the Bain strain is the optimal lattice transformation from fcc-to-bcc with respect to three of the most commonly used strain measures.

The structure of the paper is as follows: after stating some preliminaries in Section~\ref{SecPre} we explore in more depth some mathematical aspects of lattices in Section~\ref{SecMet}. This section is mainly intended for the mathematically inclined reader and may be skipped on first reading without inhibiting the understanding of Section~\ref{SecOpt}, which constitutes the main part of this paper. In this section we establish a geometric criterion of optimality and prove the existence of optimal lattice transformations for any displacive phase transition between two Bravais lattices. Additionally a precise algorithm to compute these optimal strains is provided. In the remaining subsections, the general theory is applied to prove the optimality of the Bain strain in an fcc-to-bcc transformation, to show that the Bain strain remains optimal in an fcc to body-centred tetragonal (bct) transformation and finally to derive the optimal transformation strain between two triclinic phases 
of Terephthalic Acid. Similarly to the fcc-to-bcc transition, this phase transformation is of particular interest as it involves large stretches and thus the lattice transformation requiring least atomic movement is not clear.      
\newpage
\section{Preliminaries}\label{SecPre}
Throughout it is assumed that both the parent and product lattices are Bravais lattices (see Definition~\ref{DefBrav}) and that the transformation strain, i.e. the deformation that maps a unit cell of the parent lattice to a unit cell of the product lattice, is homogeneous. 

The following definitions are standard and will be used throughout. 
\begin{Def}
Let $\mathcal{R}\in \lbrace \mathbb{Z},\mathbb{R}\rbrace$ denote the set of integer or real numbers respectively and define
\begin{itemize}
 \item[\icon]  $\mathcal{R}^{3 \times 3}$: Vector space of matrices with entries in $\mathcal{R}$.
 \item[\icon]  $\mathcal{R}^{3 \times 3}_+\colonequals \mathcal{R}^{3 \times 3} \cap \lbrace \det >0\rbrace$: Orientation preserving matrices with entries in $\mathcal{R}$. 
  \item[\icon]  $\mathrm{GL}(3,\mathcal{R})\colonequals \mathcal{R}^{3 \times 3} \cap \lbrace A\in  \mathcal{R}^{3 \times 3} \mbox{ is invertible}\rbrace$ \emph{(General linear group)}.
\item[\icon]   $\mathrm{GL}^+(3,\mathcal{R})\colonequals \mathrm{GL}(3,\mathcal{R}) \cap \mathcal{R}^{3 \times 3}_+$: Set of orientation preserving invertible $3\times 3$ matrices with entries in $\mathcal{R}$.
\item[\icon]  $\mathrm{SL}(3,\mathbb{Z})\colonequals \mathrm{GL}^+(3,\mathbb{Z})=\lbrace A \in \mathbb{Z}^{3\times 3}: \det A=1\rbrace$ \emph{(Special linear group)}.
\item[\icon]  $\SO=\lbrace R\in \rr: R\T R=\id, \det R=1\rbrace$ \emph{(Group of proper rotations)}.
\item[\icon]  $\PP \subset  \SO$ \emph{(Symmetry group of a cube)} - see Lemma~\ref{lemmap24} below.
\end{itemize}
Further define the multiplication of a matrix $F$ and a set of matrices $\mathcal{S}$ by $F.\mathcal{S}\colonequals\lbrace FS: S\in \mathcal{S}\rbrace$. 
\end{Def}
The following Lemma establishes a characterisation of the group $\PP$, i.e. the set of all rotations that map a cube to itself.
\begin{Lemma}\label{lemmap24}
Let $Q = [-1,1]^3$ be the cube of side length 2 centred at $0$ and define $\PP = \{P\in \SO : PQ=Q\}$. Then $|\PP|=24$ and
 $\PP = \SO \cap \Sl$. 
\end{Lemma}
\begin{proof}
Suppose that $P\in\PP$ and let $\{e_i\}_{i=1,2,3}$ denote the standard basis of $\ro^3$. By linearity, $P$ maps the face centres of $Q$ to face centres, i.e. for each $i=1,2,3$ there exists $j\in \{1,2,3\}$ such that
 $Pe_i=\pm e_{j}$.
Therefore, $P_{ki}=Pe_i\cdot e_k = \pm\delta_{kj}\in \lbrace {-1,0,1}\rbrace$ and thus $P\in \slz$.

Conversely, if $P\in \SO \cap \slz$, its columns form an orthonormal basis and since $P$ has integer entries the columns have to be in the set $\{\pm e_i\}_{i=1,2,3}$. Hence, $P$ is of the form
\[
P=\sum_{j=1}^3 \pm e_{k_j} \otimes e_j
\] 
and $Pe_i = \pm (e_i\cdot e_{j}) e_{k_j}$. Thus, $P$ maps face centres of $Q$ to face centres and, by linearity, the cube to itself. Further, there are precisely six choices (3$\times$2) for the first column $\pm e_{k_1}$, four choices (2$\times$2) for $\pm e_{k_2}$ and two choices for $\pm e_{k_3}$. Thus, taking into account the determinant constraint, there are $24$ elements in $\PP$.
 \end{proof}
 
\begin{Remark}
 Essentially the same proof can be used to show that $\mathrm{SO(\mathnormal{n})}\cap{{\mathrm{SL(\mathnormal{n},\mathbb{Z})}}}={\mathcal{P}}^{\mathnormal{N_n}}$, where 
$\mathcal{P}^{N_n}$ with $N_n=2^{n-1}n!$ denotes the symmetry group of a $n$-dimensional cube. As before the value of $N_n$ arises from having $2 \times n$ choices for the first column of $Q$, then $2 \times (n-1)$ choices for the second column of $Q$, ... and the last column of $Q$ is already fully determined by the determinant constraint.   
\end{Remark}

\begin{Def}{(Pseudometric and metric)}\\
Let $X$ be a vector space and $x$, $y$, $z \in X$. A map $\dm: X \times X \ra [0,\infty)$ is a \emph{pseudometric} if
\begin{enumerate}
 \item  $\dm(x,x) = 0$, \label{m1} 
 \item  $\dm(x,y)=\dm(y,x)$ \emph{(symmetry)},\label{m2}
 \item  $\dm(x,z) \leq \dm(x,y)+\dm(y,z)$ \emph{(triangle inequality)}.\label{m3}
\end{enumerate}
If in addition $\dm$ is positive definite, i.e. $\dm(x,y) = 0 \Lra x=y$, we call $\dm$ a \emph{metric}.
\end{Def}

\begin{Def}{(Matrix norms)}\\ \label{DefMatrixNorms}
For a given matrix $A \in \rr$ we define the following norms
\begin{itemize}
 \item[\icon]  \emph{Frobenius norm}: $$|A|=\sqrt{\Tr (A\T A)}=\sqrt{\sum_{i,j=1}^3 A_{ij}^2}.$$
 \item[\icon]  \emph{Spectral norm}: $$|A|_2=\sup_{|x|=1}|Ax|=\sqrt{\max_{i=1,2,3}{\lambda_i(A\T A)}}= \max_{i=1,2,3}{\nu_i(A)},$$where for $i=1,2,3$, $\nu_i(A)$ are the principal stretches/singular values of $A$ \\ \hspace*{1cm}and $\lambda_i(A\T A)$ are the eigenvalues of $A\T A$.
 \item[\icon]  \emph{Column max norm}: $$\|A\|_{2,\infty}=\max_{i=1,2,3} |Ae_i|=\max_{i=1,2,3} |a_i|,$$ where $\lbrace a_1,a_2,a_3\rbrace$ are the columns of $A$.
\end{itemize}
Unless otherwise specified, here and throughout the rest of the paper $|\cdot|$ always denotes the Euclidean norm if the argument is a vector in $\ro^3$ and the Frobenius norm if the argument is a matrix in $\rr$. Additionally, we henceforth denote by $\col{A}\colonequals\lbrace a_1,a_2,a_3\rbrace$ the columns of the matrix $A$ and then write $A=[a_1,a_2,a_3]$.
\end{Def}
The proofs of the following statements are elementary and can be found in standard textbooks on linear algebra (see e.g. \cite{MatrixBook}).
\begin{Lemma}{(Properties of matrix norms)}\label{LemFrob}\\
Both the Frobenius norm and the spectral norm are \emph{unitarily equivalent}, that is
\begin{equation}
 |RAS|=|A| \label{EqCy}
\end{equation}
for any $R,S \in \SO$. Further both norms are \emph{sub-multiplicative}, that is given $A,B,C \in \rr$ then $|ABC|\leq |A||B||C|$ and thus in particular if $|A||C|\neq0$ then
 \begin{equation}\label{EqSubM}
  |B|\geq \frac{|ABC|}{|A||C|}.
 \end{equation}
 Further the spectral norm is compatible with the Euclidean norm on $\ro^3$, that is 
 \begin{equation}\label{Eq2normComp}
  |Ax| \leq |A|_2 |x|.
 \end{equation}
 \end{Lemma}
The following sets will be of particular importance when proving the optimality of lattice transformations.
\begin{Def}{($\Slk$)} \label{Defslk} \\
 For $k\in \mathbb{N}$ define
 \begin{equation}\no
 \Slk \colonequals \lbrace A \in \Sl: |A_{mn}|\leq k\ \forall \, m,n\in \lbrace 1,2,3\rbrace \rbrace
\end{equation}
and
 \begin{equation}\no
 {{\mathrm{SL^{-k}(3,\mathbb{Z})}}} \colonequals \lbrace A \in \Sl: |\left(A^{-1}\right)_{mn}|\leq k\ \forall\,  m,n\in \lbrace 1,2,3\rbrace \rbrace.
\end{equation}
Clearly ${{\mathrm{SL^j(3,\mathbb{Z})}}}\subset {{\mathrm{SL^k(3,\mathbb{Z})}}} \ \forall\, 0 \leq j \leq k$ and $|{{\mathrm{SL^{-k}(3,\mathbb{Z})}}}|=|\Slk|$ for all $k\in \mathbb{Z}$. For example we have $|\SL{1}|=3\,480$, $|\SL{2}|=67\,704$, $|\SL{3}|=640\,824$, $|\SL{4}|=2\,597\,208$, $|\SL{5}|=10\,460\,024$ and $|\SL{6}|=28\,940\,280$.

 \end{Def}
 Below we recall some basic definitions and results from crystallography.
\begin{Def}{(Bravais lattice, \cite{Bha} Ch. 3)} \label{DefBrav} \\
Let $F=[f_1,f_2,f_3]\in \glp$, where $\col{F}=\lbrace f_1,f_2,f_3\rbrace $ are the columns of $F$. We define the \emph{Bravais lattice $\mathcal{L}(F)$ generated by $F$} as the lattice generated by $\col F$, i.e.
\begin{equation}\no
 \mathcal{L}(F)\colonequals \col{F. \zp}.
\end{equation}
Thus by definition a Bravais lattice is $\spann_{\mathbb{Z}}\lbrace f_1,f_2,f_3 \rbrace$ together with an orientation.
\end{Def}
\begin{Def}{(Primitive, base-, body- and face-centred unit cells)}\\
Let $\ml=\mathcal{L}(F)$ be generated by $F\in \glp$. We call the parallelepiped spanned by $\col F$ with one atom at each vertex a \emph{primitive unit cell} of the lattice. We call a primitive unit cell with additional atoms in the centre of the bases a \emph{base-centred} unit cell, we call a primitive unit cell with one additional atom in the body centre a \emph{body-centred} (bc) unit cell and we call a primitive unit cell with additional atoms in the centre of each of the faces a \emph{face-centred} (fc) unit cell.  
\end{Def}
\begin{Remark}
For any lattice generated by a base-, body- or face-centred unit cell there is a primitive unit cell that generates the same lattice. The following table gives the lattice vectors that generate the equivalent primitive unit cell for a given base-centred (C), body-centred (I) or face-centred (F) unit cell spanned by the vectors $\{a,b,c\} \in \R$.
\begin{table}[h]
\begin{center}
  \begin{tabular}{c*{4}{c}}
primitive (P) & base-centred (C) & body-centred (I) & face-centred (F)& \\
\hline
$\{a,b,c\}$ & $\{\frac{a-b}{2},\frac{a+b}{2},c\}$ & $\{\frac{-a+b+c}{2},\frac{a-b+c}{2},\frac{a+b-c}{2}\}$ & $\{\frac{b+c}{2},\frac{a+c}{2},\frac{a+b}{2}\}$\\
\end{tabular}
\end{center}
\caption{Lattice vectors of a primitive unit cell that generates the same lattice.}\label{TableUnit}
\end{table}
For our purposes all unit cells that generate the same lattice are equivalent and in order to keep the presentation as simple as possible we will always work with the primitive description of a lattice. However, we note that often in the literature the unit cell is chosen such that it has maximal symmetry. 

For example for a face-centred cubic lattice, the unit cell would be chosen as face-centred and spanned by $\col{\id}=\lbrace e_1,e_2,e_3\rbrace$ so that it has the maximal $\PP$ symmetry. However, if one considers primitive unit cells that span the same fcc lattice, the one with maximal symmetry is given by the last entry in Table~\ref{TableUnit} and thus spanned by $\col{F}$, where  
\begin{equation} \no
 F=\frac{1}{2}\begin{pmatrix} 0 & 1 & 1 \\ 1 & 0 & 1\\ 1 & 1 & 0 \end{pmatrix}
\end{equation}
and has only $6$-fold symmetry.
\end{Remark}
\begin{Lemma}{(Identical lattice bases, \cite{Bha} Result 3.1)}\label{LemEquiv} \\
Let $\ml(F)$ be generated by $F=[f_1,f_2,f_3]$ and $\ml(G)$ be generated by $G=[g_1,g_2,g_3]$. Then 
\begin{equation}\no
 \mathcal{L}(F)=\mathcal{L}(G)   \Lra  G=F\mu \Lra g_i = \sum_{j=1}^3\mu_{ji}f_j, 
\end{equation}
for some $\mu\in\slz$.
The same result holds for a face-, base- and body-centred unit cell.
\end{Lemma}
 \begin{Def}{(Lattice transformation)}\label{def:lattice_transformation}\\
Given two lattices $\mathcal{L}_0=\mathcal{L}(F)$ and $\mathcal{L}_1=\mathcal{L}(G)$ generated by $F,G \in \glp$ we call any matrix $H\in \glp$ such that $H.\mathcal{L}_0=\mathcal{L}_1$ a \emph{lattice transformation} from $\mathcal{L}_0$ to $\mathcal{L}_1$.
\end{Def}
\begin{Remark}\label{RemPoint}
In the terminology of Ericksen (see e.g. \cite{Zanz} p. 62ff.), if $\ml_0=\ml_1$ (i.e. $G=F\mu$), the matrices $H$ in Definition~\ref{def:lattice_transformation} are precisely all the orientation preserving elements in the \emph{global symmetry group} of $F$. Additionally, the matrices $H$ that are also rotations constitute the \emph{point group} of the lattice. We point out that, in this terminology, $\PP$ is the point group of any cubic lattice.
\end{Remark}

We end this section by defining the atom density of the lattice $\ml(F)$ and relating it to the determinant of $F$.
\begin{Def}{(Atom density)}\label{DefAtom}\\
 For a given lattice $\ml$ we define the atom density $\rho(\ml)$ by
 \begin{equation}\no
  \rho(\ml)\colonequals \lim_{N\raf} \frac{\#\lbrace{Q_N\cap \ml}\rbrace}{N^3},
 \end{equation}
 where $Q_N=[0,N]^3$ is the cube with side-length $N$ and $\#$ counts the number of elements in a discrete set. Thus $\rho(\ml)$ is the average number of atoms per unit volume.
\end{Def} 
\begin{Lemma}
 Let $\ml=\ml(F)$ be generated by $F\in \glp$ then
 \begin{equation}\no
  \rho(\ml)=\frac{1}{\det F}.
 \end{equation}
In particular a transformation $H=GF^{-1}$ from $\ml_0=\ml(F)$ to $\ml_1=\ml(G)$ does not change the atom density if and only if it is volume preserving, i.e. $\det H=1$. Further the atom density is well defined.
\end{Lemma}
\begin{proof}
 Denote by $\mathcal U \subset \ro^3$ the unit cell spanned by $\col F$, so that the volume of $\mathcal U$ is given by $|\mathcal U|=\det F$. Taking $n$ distinct points $x_i \in \ml$ we find that $|\bigcup_{i=1}^n (x_i+\mathcal{U})|=$ $n \det F$ since all elements are disjoint (up to zero measure). Let $l$ denote the side length of the smallest cube that contains $\mathcal U$. Further, as in Definition~\ref{DefAtom}, let $Q_N=[0,N]^3$ denote a cube of side-length $N$ and define $Q^\pm_N=[\mp2l,N\pm2l]^3$. Then
 \begin{equation}\no
  Q^-_N \subset \bigcup_{x \in \ml \cap Q_N}(x+\mathcal{U}) \subset Q^+_N
 \end{equation}
and thus by taking the volumes of the sets we obtain
\begin{equation}\no
 (N-4l)^3 \leq\#\lbrace{Q_N\cap \ml} \rbrace \det F\leq (N+4l)^3.
\end{equation}
Dividing by $N^3$ and taking the limit $N\raf$ yields the result. Since this limit exists for all sequences $N\raf$ the density is well defined.
\end{proof}

\section{Metrics and equivalence on matrices and lattices}\label{SecMet}

\begin{Def}{(Equivalent matrices and lattices )}\\
We define an equivalence relation $\thicksim$ between matrices $F,G$ in $\glp$ by
\begin{equation} \no
 F \thicksim G :\Lra \exists R \in \SO : G=RF,
\end{equation}
so that the equivalence class $[F]$ of $F$ is given by $[F]=\lbrace G \in \rr_+ : F \thicksim G \rbrace$. We denote the quotient space, i.e. the space of all equivalence classes, by 
\begin{equation} \no
 \overline{\glp} \colonequals \lbrace [F] : F \in \glp \rbrace.
\end{equation}
Furthermore, we define an equivalence relation $\thicksim$ between two lattices $\ml_0$ and $\ml_1$ by
\begin{equation}\no
 \ml_0 \thicksim \ml_1 :\Lra \exists R\in \SO: \ml_1 = R.\ml_0.
\end{equation}
\end{Def}
We are now in a position to define a metric on the quotient spaces. 

\begin{LemDef}{(Induced metric)} \label{LemmaInd}\\
Any pseudometric $\dm:\glp \times \glp \ra [0,\infty)$ with the property 
\begin{equation}\label{metriczero}
 \dm(F,G)=0 \Lra G=T^*F  \mbox{ for some } T^* \in \SO \tag{$\ast$}
\end{equation}
naturally induces a metric $\odm$ on $\overline{\glp}$ via
\begin{equation} \no
 \odm([F],[G])=\min_{R, S \in \SO}\dm (RF,SG).
\end{equation}
\end{LemDef}
\begin{proof}
The quantity $\odm$ is clearly well defined, so that in particular $\odm([F],[G])=\odm (F,G)$ and we may henceforth drop the $[\cdot]$ in the arguments of $\odm$. We first show that $\odm$ is a metric on $\overline{\glp}$. Positivity and symmetry are obvious from the definition. For definiteness note that if $\min_{R, S \in \SO}\dm (RF,SG)=0$ then by \eqref{metriczero} we have $S^*G=T^*R^*F$ for some $R^*, S^*, T^* \in \SO$ and thus $F\thicksim G$. It remains to show the triangle inequality. We have
 \begin{alignat}{3} \no
 \odm(F,H)&=\min_{R, S\in \SO}\dm (RF,SH) \\ \no
 &\leq  \min_{R,S} \left(\dm (RF,G)+\dm (G,SH)\right) \\ \no
 &\leq \min_{R,S'}\left(\dm (RF,S'G)+ \cancel{\dm(G,S'G)}\right)+\min_{R',S}\left(\dm (SH, R'G)+\cancel{\dm(G,R'G)}\right)\\ \no
 &=\odm(F,G)+\odm(G,H),\no
  \end{alignat}
where we have used the triangle inequality, symmetry of $\dm$ and \eqref{metriczero}.
\end{proof}

\begin{Example}\label{ExDi}
The family of maps $\dm_r:\glp \times \glp\ra [0,\infty), \, r \in \ro \backslash \{0\}$ given by
\begin{equation}\no
 \dm_r(F,G)=|(F\T F)^{r/2}-(G\T G)^{r/2}|
\end{equation}
are pseudometrics such that \eqref{metriczero} holds. In particular, each of them induces a metric $\odm_r$ on the quotient space $\overline{\glp}$.
\end{Example}
\begin{proof}
Positivity is obvious. The triangle inequality follows from the corresponding property of the Frobenius norm, i.e.
 \begin{alignat*}{3}
  \dm_r(F,H)&=|(F\T F)^{r/2}-(H\T H)^{r/2}|\\
  &\leq |(F\T F)^{r/2}-(G\T G)^{r/2}| +|(G\T G)^{r/2}-(H\T H)^{r/2}|=\dm_r(F,G)+\dm_r(G,H)
 \end{alignat*}
 and the property \eqref{metriczero} follows from
 \begin{alignat*}{3}
   \dm_r(F,G)=0 &\Lra (F\T F)^{r/2}=(G\T G)^{r/2} \Lra (FG^{-1})\!\T (FG^{-1})=\id \\
   &\Lra FG^{-1}\in \SO \Lra F=T^*G \mbox{ for some } T^* \in \SO.
 \end{alignat*}
\end{proof}
\begin{Remark}
The metric $\dm_2$ has previously been used in \cite[Chapter 3]{Bha}, where it was defined as the distance between the metric $C_F=F\T F=(f_i \cdot f_j)_{ij}$ of a set of lattice vectors $\col{F}=\lbrace {f_1,f_2,f_3} \rbrace$ and the metric $C_G=G\T G=(g_i \cdot g_j)_{ij}$ of a set of lattice vectors $\col{G}=\lbrace {g_1,g_2,g_3}\rbrace$. The use of the term metric in \cite{Bha} is not to be confused with the use of metric in the present paper.   
\end{Remark}

\section{Optimal lattice transformations}\label{SecOpt}
This section embodies the main part of the present paper. We first establish what we mean by an optimal transformation from one lattice to another and then, for a family of such criteria, show the existence of optimal transformations between \emph{any} two Bravais lattices. 
\begin{Lemma}\label{LemmaLat}
 Let $\mathcal{L}_0=\mathcal{L}(F)$ and $\mathcal{L}_1=\mathcal{L}(G)$ be generated by $F,G \in \glp$. Then all possible lattice transformations from $\mathcal{L}_0$ to $\mathcal{L}_1$ are given by $H_\mu=G\mu F^{-1}, \mu \in \slz$. In particular, the lattices coincide if and only if there exists $\mu \in \slz$ such that $G\mu F^{-1}=\id$ and they are equivalent, i.e. $\mathcal{L}_0 \thicksim \mathcal{L}_1$, if and only if there exists $\mu \in \slz$ such that $G\mu F^{-1}\in \SO$. 
\end{Lemma}
\begin{proof}
Let $H_\mu=G\mu F^{-1}, \mu \in \slz$. Then 
\begin{alignat*}{3}
  H_\mu . \mathcal{L}_0&=H_\mu .\mathcal{L}(F)=H_\mu .\col{F.\zp}=\col{H_\mu F.\zp}\\
  &=\col{G\mu F^{-1}F.\zp}=\col{G\mu.\zp}=\col{G.\zp}=\mathcal{L}_1,
 \end{alignat*}
where we used that $\mu$ is invertible over $\mathbb{Z}$, so that $\mu.\zp=\zp$. Thus $H_\mu$ is a lattice transformation from 
$\mathcal{L}_0$ to $\mathcal{L}_1$.
For the reverse direction we know by Lemma~\ref{LemEquiv} that $\mathcal{L}{(F)}=\mathcal{L}{(F')}$ if and only if $F'=F\mu', \mu' \in \slz$ and $\mathcal{L}{(G)}=\mathcal{L}{(G')}$ if and only if $G'=G\mu'', \mu'' \in \slz$ so that all possible generators for $\mathcal{L}_0$ are given by $F\mu', \mu' \in \slz$ and all possible generators for $\mathcal{L}_1$ are given by $G\mu'', \mu'' \in \slz$. Thus any lattice transformation from $\mathcal{L}_0$ to $\mathcal{L}_1$ is given by
\begin{equation}\no
 H_{\mu'}^{\mu''}=G\mu''{\mu'}^{-1}F^{-1},\ \ \mu',\mu''\in \slz.
\end{equation}
But by the group property we may set $H_\mu\colonequals H_{\mu'}^{\mu''}$ with $\mu=\mu''{\mu'}^{-1}\in \slz$.
\end{proof}

\begin{Def}{($\dm$-optimal lattice transformations)}\label{DefOpti}\\
 Given two lattices $\mathcal{L}_0=\mathcal{L}(F)$ and $\mathcal{L}_1=\mathcal{L}(G)$ with $F,G \in \glp$ and a metric $\dm$, we call a lattice transformation $H_\mu= G\mu F^{-1}, \ \mu \in \slz$ (cf. Lemma~\ref{LemmaLat}) \emph{$\dm$-optimal} if it minimises the distance to $\id$ with respect to $\dm$, 
 i.e. $H_{\min}=G\mu_{\min} F^{-1}$, where 
 \begin{alignat}{3}\label{muMin}
   \mu_{\min}=\argmin_{\mu\in \slz} \dm({H_\mu},\id). 
 \end{alignat}
If $\dm$ is a pseudometric satisfying \eqref{metriczero}, we call a transformation \emph{$\dm$-optimal} if it is optimal with respect to the induced metric $\odm$ on the quotient space $\overline{\glp}$, i.e. the $\dm$-optimal transformation is the one that is $\dm$-closest to being a pure rotation and maps $\ml_0$ to a lattice in the equivalence class $[\ml_1]:=\lbrace \ml: \ml \thicksim \ml_1\rbrace$ of $\ml_1$. 
 \end{Def}
%
\begin{Example}
For the pseudometrics $\dm_r$ from Example~\ref{ExDi} the explicit expressions for the distance in \eqref{muMin} read
\begin{equation}
 \dm_r(H,\id)=|(H\T H)^{r/2}-\id|=\left( \sum_{i=1}^3 (\nu_i^r-1)^2\right)^{1/2},\label{Hdr}
\end{equation}
where $\nu_i, \, i=1,2,3$ are the principal stretches/singular values of $H$. The quantities $(H\T H)^{r/2}-\id$ are clearly measures of principal strain and are known as the Doyle-Ericksen strain tensors (see \cite[p. 65]{DoyleErick}). For $r=1$, it is simple to verify that
\begin{equation}\no
 \dm_1(H,\id)=\dist(H,\SO)=\min_{R \in \SO} |H-R|.
 \end{equation}
\end{Example}
 \begin{Remark} 
  For a general metric $\dm$ as in Definition ~\ref{DefOpti}, the optimal transformation $H_{\min}$ between $\mathcal{L}(F)$ and $\mathcal{L}(G)$ is unchanged under actions of the point groups of both lattices, i.e.  
  \begin{equation}\nonumber
   H_{\min}=G\argmin_{\mu\in \slz} \dm({G\mu F^{-1}},\id) F^{-1} = (PG)\argmin_{\mu\in \slz} \dm({(PG)\mu (QF)^{-1}},\id) (QF)^{-1} 
  \end{equation}
for any $P$ in the point group of $\ml(G)$ and $Q$ in the point group of $\ml(F)$. Throughout the rest of the paper, we only use pseudometrics satisfying \eqref{metriczero}. In this case, the notion of optimality is not only invariant under actions of the respective point groups but also under rigid body rotations of the product lattice. Thus Definition~\ref{DefOpti} returns an equivalence class $[H_{\min}]\in \overline{\glp}$ of $\dm$-optimal transformations. By the polar decomposition theorem we may, and henceforth always will, pick the symmetric representative
 \begin{equation}\no
  \bar{H}_{\min}:=\sqrt{H_{\min}\T H_{\min}} \in \rr_{\mbox{\scriptsize{sym}}},
 \end{equation}
 i.e. the pure \emph{stretch component} of the transformation $H_{\min}$. Note that in general the set of minimising equivalence classes $\lbrace [H_{\min}] :  H_{\min} \mbox{ is $\dm$-optimal}\rbrace$ may contain more than one element. In such a case, different regions of the parent lattice may transform according to any of these optimal strains, giving rise to e.g. twinning.
 
The pseudometrics $\dm_r$ are additionally invariant under rotations from the right, i.e. $\dm_r(H,\id)=\dm_r(H S,\id)$ for all $S\in \SO$. For any such pseudometric, a rigid body rotation $R$ of the parent lattice $\ml_0$ results in an optimal transformation with stretch component $R \bar{H}_{\min}R\T $, where $\bar H_{\min}$ is the stretch component of the optimal transformation from $\ml_0$ to $[\ml_1]$. Note that $R \bar{H}_{\min}R\T $ is simply $\bar{H}_{\min}$ expressed in a different basis and in particular, even though the coordinate representation is different, the underlying transformation mechanism is unchanged. 
 \end{Remark}
Our main theorem says that a $\dm_r$-optimal lattice transformation always exists and the following lemma will be the crucial tool in the proof.

\begin{KeyLemma}\label{LemmaKey}
Let $H$ be a lattice transformation from $\ml_0=\ml(F)$ to $\ml_1=\ml(G)$ and consider a lattice vector $f\in \ml_0$ that is transformed by $H$ to $g=Hf$. Then 
 \begin{equation} \no
  \nu_{\max}(H) \geq |g|/|f| \geq  \nu_{\min}(H),
 \end{equation}
where $\nu_{\min}(H)$, $\nu_{\max}(H)$ denote, respectively, the smallest and largest principal stretches/singular values of $H$. In particular for any $s>0$,
\begin{alignat}{3}
  \dop_s(H,\id)&=\left(\sum_{i=1}^3(\nu_i^s-1)^2\right)^{1/2} &\geq \max_i \frac{|Hf_i|^s}{|f_i|^s}-1, \label{d2bound} \\
    \dop_{-s}(H,\id)&=\left(\sum_{i=1}^3(\nu_i^{-s}-1)^2\right)^{1/2} &\geq \max_i\frac{|H^{-1}g_i|^s}{|g_i|^s}-1, \label{d-2bound}
\end{alignat}
where $\col{F}=\lbrace{f_1,f_2,f_3}\rbrace$ and $\col{G}=\lbrace g_1,g_2,g_3\rbrace$.
 \end{KeyLemma}
\begin{proof}
 Consider the singular value decomposition of $H=UDV$, where $D=\diag(\nu_1,\nu_2,\nu_3)$ and $U,V \in \SO$. Then
 \begin{equation} \no
  |g|=|UDVf|=|DVf|\leq \max_i\,{\nu_i(H)}|Vf|= \nu_{\max}(H)|f| 
 \end{equation}
and analogously for the lower bound.
\end{proof}

\begin{Theorem}\label{TheCpct}
 Given two lattices $\mathcal{L}_0=\mathcal{L}(F)$ and $\mathcal{L}_1=\mathcal{L}(G)$ generated by $F,G \in \glp$ respectively, there exists a $d_r$-optimal lattice transformation $H_{\mu_{\min}}=G{\mu_{\min}} F^{-1}$, for any $r \in \ro \backslash \{0\}$. For $s>0$ all optimal changes of basis are contained in the finite compact sets 
\begin{alignat}{3}
 \dm_s: \quad &\left \lbrace \mu\in \slz : \tinf{\mu}^s\leq \frac{\tinf{F}^s}{\nu_{\min}^s(G)}(m_{0,s}+1)\right\rbrace, \label{Eqmu2} \\
  \dm_{-s}: \quad &\left \lbrace \mu\in \slz : \tinf{\mu^{-1}}^s\leq \frac{\tinf{G}^s}{\nu_{\min}^s(F)}(m_{0,-s}+1)\right\rbrace, \label{Eqmu-2} 
\end{alignat}
where $\nu_{\min}(A)$ denotes the smallest principal stretch/singular value of $A$ and $$m_{0,r}\colonequals \dm_{r}({H_\id},\id)=\dm_{r}(GF^{-1},\id).$$
\end{Theorem}
\begin{proof}
As the minimisation is over the discrete set $\slz$ it suffices to show that the minimum is attained in a (compact) finite subset of $\slz$ given by \eqref{Eqmu2} and \eqref{Eqmu-2}. Let $H_\mu=G \mu F^{-1}, \, \mu \in \slz$ be a lattice transformation from $\ml_0=\ml(F)$ to $\ml_1=\ml(G)$. Then, letting $\{e_i\}_{i=1,2,3}$ denote the standard basis vectors of $\rr$,
\begin{equation} \no
 H_\mu f_i=G \mu F^{-1}Fe_i=G \mu e_i \aand H_{\mu}^{-1}g_i=F\mu^{-1} G^{-1}Ge_i=F\mu^{-1} e_i
\end{equation}
and thus, by using the Key Lemma~\ref{LemmaKey} and \eqref{Eq2normComp}, we obtain
\begin{alignat*}{4}
  \dop_s(H_\mu,\id) &\geq \max_i \frac{|G \mu e_i|^s}{|f_i|^s}-1 &\geq \frac{\|\mu\|^s_{2,\infty}}{|G^{-1}|_2^s\|F\|^s_{2,\infty}}-1 &=\frac{\nu_{\min}^s(G)\|\mu\|^s_{2,\infty}}{\|F\|^s_{2,\infty}}-1,\\
    \dop_{-s}(H_\mu,\id) &\geq \max_i\frac{|F\mu^{-1} e_i|^s}{|g_i|^s}-1 &\geq \frac{\|\mu^{-1}\|^s_{2,\infty}}{|F^{-1}|_2^s\|G\|^s_{2,\infty}}-1 &=\frac{\nu_{\min}^s(F)\|\mu^{-1}\|^s_{2,\infty}}{\|G\|^s_{2,\infty}}-1, 
\end{alignat*}
where in the equality we have used that $\lbrace \nu_i (A^{-1})\rbrace_{i=1,2,3}=\lbrace (\nu_i(A))^{-1}\rbrace_{i=1,2,3}$. Thus $\dm_r({H_\mu},\id)> \dm_r({H_\id},\id)$ for all $\mu$ in the complement of the respective sets given by \eqref{Eqmu2} and \eqref{Eqmu-2} and therefore $H_\mu$ cannot be $\dm_r$-optimal. 
\end{proof}

\begin{Remark} \label{RemarkDistances}
The distance $\dm_1(H,\id)$ seems to be the most natural candidate to determine the transformation requiring least atomic movement. The quantities $\nu_i-1$ measure precisely the displacement along the principal axes and thus their use is in line with the criterion of smallest principal strains as in e.g. \cite{Bhadeshia}, \cite{BowlesWayman} and \cite{Lomer}. The distance $\dm_2(H,\id)$ seems natural from a mathematical perspective as the tensor $H\T H$ corresponds to the flat metric induced by the deformation $H$ and it has also been used to define the Ericksen-Pitteri neighbourhood of a lattice (see e.g. $(2.17)$ in \cite{BallTests}). Finally, the distance $\dm_{-2}(H,\id)$ has recently been used in \cite{Sherry} in order to avoid singular behaviour when considering sublattices.
\end{Remark}
Below we illustrate the differences of $\dm_1, \dm_2$ and $\dm_{-2}$ through a simple but instructive 1D example.
\begin{Example}{(A comparison of different optimality conditions)}\\
 We consider two atoms $A,B$ that are originally at unit distance, i.e. $|A-B|=1$ and then move the atom $B$ to its deformed position $B'$. Thus $H$ is simply a scalar quantity given by $H=|A-B'|/|A-B|=|A-B'|$.
      \begin{table}[h]
\begin{center}
  \begin{tabular}{{r}{l}{l}{c}}
 & $B'-B=0.5,$ &$B'-B=-0.5,$ & deformation $y$ such that  \\
 r &  $H=|A-B'|=1.5$ &  $H=|A-B'|=0.5$  & $\dm_r(y,\id)=\dm_r(x,\id)$ \\
\hline
1 & $\dm_1(H,\id)=0.5$ & $\dm_1(H,\id)=0.5$ & $y=2-x$  \\
2 & $\dm_2(H,\id)=1.25$& $\dm_2(H,\id)=0.75$ &$y=\sqrt{2-x^2}$ \\
-2 & $\dm_{-2}(H,\id)=0.\bar5$  & $\dm_{-2}(H,\id)=3$ & $y=\frac{1}{\sqrt{2-x^{-2}}}$ 
\end{tabular}
\end{center}
\caption{Comparison of different distances} \label{TableCompDist}
\end{table}
 It can be seen from Table~\ref{TableCompDist} that $\dm_1$ depends only on the distance between $B$ and $B'$; an expansion by $100\%$ has the same $\dm_1$ distance to $\id$ as a contraction to $0$, i.e. moving $A$ onto $B$. The metric $\dm_2$ penalises expansions more than contractions; e.g. an expansion by $\approx~141\%$ has the same $\dm_2$ distance to $1$ as a contraction to $0$. The metric $\dm_{-2}$ penalises contractions significantly more than expansions; e.g. an expansion by $\infty$ has the same $\dm_{-2}$ distance to $\id$ as a contraction to $\approx 70\%$, i.e. reducing the distance between $A$ and $B$ by $\approx 30\%$. 
\end{Example}

\subsubsection*{A remark on the computation of the optimal transformation}
Theorem~\ref{TheCpct} provides the necessary compactness result to reduce the original minimisation problem over the infinite set $\slz$ to a finite subset given by \eqref{Eqmu2} and \eqref{Eqmu-2} respectively. To this end, it is worth noting that the smaller the deformation distance $m_{0,r}=\dm_r({GF^{-1}},\id)$ of the initial lattice basis the smaller the radius of the ball in $\slz$ that contains the optimal $\mu$. However, in specific cases, where better estimates are available, it might be advantageous to start with an initial lattice basis that is not optimal.

Nevertheless, in order to explicitly determine the optimal transformations one still needs to compare the distances $\dm_r(H_\mu,\id)$ for all elements contained in the finite sets given by \eqref{Eqmu2} and \eqref{Eqmu-2} respectively. This can easily be carried out with any modern computer algebra program and possible implementations can be found in the \hyperref[append]{Appendix}.

In order to ensure that the solution of this finite minimisation problem is correct one needs to verify that the difference $\Delta$ between the minimal and the second to minimal deformation distance is large compared to possible rounding errors (if any). The computations in Sections~\ref{subsecBain} and~\ref{subsecBainStab} for the Bain strain from fcc-to-bcc/bct are exact and thus without rounding errors. In Section~\ref{subsecTere} regarding the optimal transformation in Terephthalic Acid we find that $\Delta>0.015$ which is large compared to machine precision. 

\subsection{The Bain strain in fcc-to-bcc}\label{subsecBain}
Having established the general theory of optimal lattice transformations we apply these results to prove the optimality of the Bain strain with respect to the three different lattice metrics $\dm_r$, $r=-2,1,2$, from the previous example. In these cases we rigorously prove the optimality of the Bain strain first proposed in \cite{Bain}. 

\begin{Theorem}{(Bain Optimality)} \label{ThBain}\\
In a transformation from an fcc to a bcc lattice with no change in atom density, there are three distinct equivalence classes of $\dm_r$-optimal lattice transformations for $r=1,2,-2$. The stretch components are given by
\begin{equation}\no
 \bar H_{\min} \in \lbrace \diag(2^{-1/3},2^{1/6},2^{1/6}), \, \diag(2^{1/6},2^{-1/3},2^{1/6}), \, \diag(2^{1/6},2^{1/6},2^{-1/3})\rbrace,
\end{equation}
i.e. the three Bain strains are the $\dm_r$-optimal lattice transformations in a volume preserving fcc-to-bcc transformation for $r=1,2,-2$.
The respective minimal metric distances are  
\begin{alignat*}{2}
 m_{\min,1}&=\dm_1({ H_{\min}},\id)=\sqrt{\left(2^{-1/3}-1\right)^2+2\left(2^{1/6}-1\right)^2}\approx  0.269,\\
 m_{\min,2}&=\dm_2({ H_{\min}},\id)=\sqrt{\left(2^{-2/3}-1\right)^2+2\left(2^{1/3}-1\right)^2}\approx  0.522,\\
  m_{\min,-2}&=\dm_{-2}({ H_{\min}},\id)=\sqrt{\left(2^{2/3}-1\right)^2+2\left(2^{-1/3}-1\right)^2}\approx  0.656.
\end{alignat*}
\end{Theorem}
  \begin{figure}[h]
  \centering
  {\includegraphics[height=4cm]{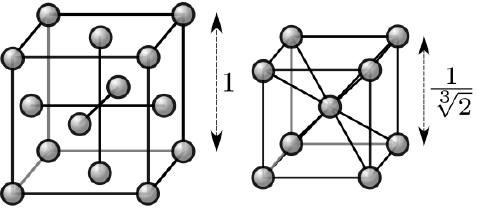}} 
  \caption{Face-centred and body-centred cubic unit cells with equal atom density.} \label{FigUnitCells}
\end{figure}
\begin{proof}
 Let $\ml_0$ denote the fcc lattice, where the fcc unit cell has unit volume and let $\ml_1$ denote the bcc lattice with the same atom density (see Fig.~\ref{FigUnitCells}). Then $\ml_0=\ml(F)$ and $\ml_1=\ml(B)$, where
 \begin{equation} \label{FB}
  F=\frac{1}{2}\begin{pmatrix} 0 & 1 & 1 \\ 1 & 0 & 1\\ 1 & 1 & 0 \end{pmatrix}=[ f_1,f_2,f_3 ], \ \ B=2^{-1/3}\frac{1}{2}\begin{pmatrix} -1 & 1 & 1 \\ 1 & -1 & 1\\ 1 & 1 & -1 \end{pmatrix}= [ b_1,b_2,b_3 ]
 \end{equation}
and, in particular, $\det F=\det B=4^{-1}$. Let $H_\mu=B \mu F^{-1}, \ \mu \in \Sl$ denote the lattice transformation from $\ml_0$ to $\ml_1$ (cf. Lemma~\ref{LemmaLat}). By definition $H_\mu$ is optimal if $\mu$ satisfies \eqref{muMin}. We first show the optimality with respect to $\dm_1$ and $\dm_2$. We may find an upper bound on the minimum by only considering $\mu \in \SL{1}$ (cf. Definition~\ref{Defslk}). With the help of a computer we find exactly $72$ such $\mu$'s and all corresponding deformations are (volume preserving) Bain strains. To complete the proof we employ our key Lemma~\ref{LemmaKey} to show that any $\mu\in \slz\backslash \SL{1}$ cannot be optimal with respect to either $\dm_1$ or $\dm_2$. Let $\mu\in \mathbb{Z}^{3 \times 3}$ be given by 
\begin{equation}\label{muAnsatz}
 \mu=\begin{pmatrix} \al_1 & \al_2 & \al_3 \\ \beta_1 & \beta_2 & \beta_3\\ \gamma_1 & \gamma_2 & \gamma_3 \end{pmatrix}.
\end{equation}
Then $ b_{\mu,i}=H_\mu f_i=B\mu e_i=\al_i b_1+\beta_i b_2+\gamma_i b_3$ and after dropping the index $i$ we obtain
\begin{equation}\label{Eqb}
 |b_\mu|^2=(\al^2+\beta^2+\gamma^2) |b_1|^2+2(\al\beta+\beta\gamma+\al\gamma) \ssp{b_1,b_2}, 
\end{equation}
where we have used that $|b_i|=|b_j|$ and $\ssp{b_i,b_j}=\ssp{b_k,b_l}$ for all $i \neq j$ and $k\neq l$. We compute $|f_i|=2^{-1/2}$, $|b_i|^2=3 \cdot 2^{-8/3}$ and $\ssp{b_i,b_j}=-2^{-8/3}$. By \eqref{d2bound} we estimate 
\begin{alignat}{3}
 \dm_1(H_\mu,\id) &\geq \frac{ 2^{-4/3}\left(\rho(\al,\beta,\gamma)\right)^{1/2}}{2^{-1/2}}-1, \label{d1FinalLB} \\
 \dm_2(H_\mu,\id) &\geq \frac{ 2^{-8/3}\rho(\al,\beta,\gamma)}{2^{-1}}-1 \label{d2FinalLB}
\end{alignat}
where $ \rho(\al,\beta,\gamma):=\al^2+\beta^2+\gamma^2+(\al-\beta)^2+(\beta-\gamma)^2+(\al-\gamma)^2$. If $\mu \in \slz\backslash \SL{1}$ then $ \rho(\al,\beta,\gamma)\geq 8$ and thus 
\begin{equation}  \no
  \dm_1(H_\mu,\id)\geq 2^{2/3}-1>0.5\gg m_{\min,1}\aand \dm_2(H_\mu,\id)\geq 2^{4/3}-1>1.5 \gg m_{\min,2}.
\end{equation}
To show $\dm_{-2}$-optimality we consider $H_{\mu}=B(F\mu^{-1})^{-1}$ and use the ansatz \eqref{muAnsatz} for $\mu^{-1}$ instead of $\mu$. We compute
\begin{equation}\no
 |F\mu^{-1} e_i|^2=\frac{1}{4}\left((\al+\beta)^2+(\beta+\gamma)^2+(\al+\gamma)^2\right)
\end{equation}
and we note that $b_i=H_{\mu}F\mu^{-1} e_i$. Thus by \eqref{d-2bound} we can estimate
\begin{alignat}{3}\label{d-2FinalLB}
 \dm_{-2}(H_\mu,\id) \geq \frac{ \frac{1}{4}\s(\al,\beta,\gamma)}{3 \cdot 2^{-8/3}}-1=2 \cdot 2^{2/3}-1 >2.17 \gg m_{\min,-2}, 
\end{alignat}
where $\s(\al,\beta,\gamma):=(\al+\beta)^2+(\beta+\gamma)^2+(\al+\gamma)^2$ and we have used that $\s(\al,\beta,\gamma)\geq 6$ for $\mu^{-1} \in \slz \backslash \SL{1}$. Therefore the $\dm_{-2}$-optimal $\mu$ is contained in $\SL{-1}$.
\end{proof}

\begin{Corollary}\label{Coralpha}
The three Bain strains remain the $\dm_r$-optimal lattice transformations, $r=1,2,-2$, from fcc-to-bcc if the volume changes by $\la^3$, provided that $\la >0.84$ for $r=1$, $\la >0.64$ for $r=2$ and $\la <1.19$ for $r=-2$. The stretch components of the three optimal equivalence classes are given by 
\begin{equation}\no
 \bar H^\la_{\min} \in \lbrace \la \diag(2^{-1/3},2^{1/6},2^{1/6}), \la \diag(2^{1/6},2^{-1/3},2^{1/6}), \la \diag(2^{1/6},2^{1/6},2^{-1/3})\rbrace.
\end{equation}
The minimal metric distances are given by  
\begin{alignat*}{2}
 m^\la_{\min,1}&=\dm_1({ H^\la_{\min}},\id)=\sqrt{\left(2^{-1/3}\la-1\right)^2+2\left(2^{1/6}\la-1\right)^2},\\
 m^\la_{\min,2}&=\dm_2({ H^\la_{\min}},\id)=\sqrt{\left(2^{-2/3}\la^2-1\right)^2+2\left(2^{1/3}\la^2-1\right)^2},\\
  m^\la_{\min,-2}&=\dm_{-2}({ H^\la_{\min}},\id)=\sqrt{\left(2^{2/3}\la^{-2}-1\right)^2+2\left(2^{-1/3}\la^{-2}-1\right)^2}.
\end{alignat*}
\end{Corollary}
\begin{proof}
 Replace $\mu\ra \la \mu$ in \eqref{muAnsatz} in the proof of Theorem~\ref{ThBain}. Then \eqref{d1FinalLB}, \eqref{d2FinalLB} and \eqref{d-2FinalLB} respectively read
 \[
\begin{array}{r@{\,}c@{\,}c@{\,}c@{\,}l@{\,}l}
 \dm_1(H^\la_\mu,\id)\geq & \frac{ 2^{-4/3}\left(\la^2\rho(\al,\beta,\gamma)\right)^{1/2}}{2^{-1/2}}-1&\geq&  \frac{2^{1/6}}{2}\la \cdot \inf_\mathcal{S} \rho^{1/2}-1,\\  
   \dm_2(H^\la_\mu,\id) \geq& \frac{ 2^{-8/3}\la^2\rho(\al,\beta,\gamma)}{2^{-1}}-1&\geq& 2^{-2/3}\la^2\cdot \inf_\mathcal{S} \rho-1,\\
  \dm_{-2}(H^\la_\mu,\id)\geq & \frac{ \frac{1}{4}\s(\al,\beta,\gamma)}{3 \cdot 2^{-8/3}\la^2}-1&\geq& \frac{2^{2/3}}{3\la^2} \cdot \inf_\mathcal{S} \s -1.
\end{array}
\]
If as above $\mathcal{S}=\slz \backslash \SL{1}$ then $\inf_\mathcal{S}\rho=8$ and thus $\dm_1(H^\la_\mu,\id)\geq m^\la_{\min,1}$ for $\la >0.84$ and $\dm_2(H^\la_\mu,\id) \geq m^\la_{\min,2}$ for $\la >0.64$ and $\inf_\mathcal{S}\s=6$ so that $\dm_{-2}(H^\la_\mu,\id)\geq m^\la_{\min,-2}$ for $\la <1.19$.
\end{proof}
The following Remark concerns the relationship between the $72$ minimising states.

 \begin{Remark}{(Relations between the minimal deformations for fcc-to-bcc)} \label{PropRel}\\
  Let $\ml(F)$ and $\ml(B)$ be the fcc and bcc lattices respectively. Let $\mu_0$ be one of the optimal changes of lattice basis and $H_i=B\mu_iF^{-1}, \quad i=0,\dotsc, 71$ denote the $72$ optimal lattice deformations associated to the optimal changes of basis $\mu_i \in \Sl$ given by Theorem~\ref{ThBain}. Then 
  all optimal $H_i$'s and corresponding $\mu_i$'s are given by
\begin{equation}\no
 H_{PQ}=PH_0Q=B\mu_{PQ}F^{-1}, \quad P,Q \in \PP,
\end{equation}
where $
 \mu_{PQ}=B^{-1}PB\mu_0F^{-1}QF\in\Sl.
$
We note that the latter equation holds since $\PP$ is the point group of both cubic lattices and thus
$B^{-1}PB$ and $F^{-1}QF$ are contained in $\slz$. Since there are only three equivalence classes of optimal lattice transformations, the 72 optimal changes of lattice basis split into three sets of 24 $\mu_{PQ}$'s such that the 24 corresponding $H_{PQ}$'s lie in the same equivalence class.
\end{Remark}
\subsection{Stability of the Bain strain}\label{subsecBainStab}
Theorem~\ref{ThBain} showed that the Bain strain is optimal in an fcc-to-bcc phase transformation. In this section, we restrict our attention to $r=1,2$ and show that the Bain strain remains optimal for a range of lattice parameters in an fcc-to-bct phase transformation. This type of transformation is found in steels with higher carbon content. The strategy of the proof is to treat the bct phase as a perturbation of the bcc phase. To this end, for $B$ as in \eqref{FB}, let the bct lattice be generated by
 \begin{equation}\label{Optbct}
  B_\ac=\diag(A,A,C)B = 2^{-4/3}   \left( \begin {array}{ccc} -A&A&A\\ \noalign{\medskip}
A&-A&A\\ \noalign{\medskip}C&C&-C
\end {array} \right),
 \end{equation}
so that $C$ denotes the elongation (or shortening) of the bcc cell in the $z$-direction and $A$ the elongation (or shortening) in the $x$- and $y$-direction. We note that, since $\PP = \slz \cap \SO$, the lattice $\ml(B_\ac)$ is equivalent to the lattices $\ml(\diag(A,C,A)B)$ and $\ml(\diag(C,A,A)B)$. Further we define $m^{0,\ac}_{i}=\dm_i(\diag(2^{1/6}A,2^{1/6}A,2^{-1/3}C),\id)$.

The following proposition provides the most important ingredient. 

\begin{Proposition}{(``The first excited state'')} \label{PropEx} \\ 
In a volume preserving transformation from an fcc-to-bcc lattice the second to minimal deformation distances are given by
\begin{equation}\no
 m_1^1\colonequals \min_{\mu\in \slz \atop \mu \neq \mu_{\min} }\dm_1 ({H_\mu},\id)\approx 0.70 \ \ \aand \ \  m_2^1\colonequals \min_{\mu\in \slz \atop \mu \neq \mu_{\min} }\dm_2 ({H_\mu},\id)\approx 1.64.
\end{equation}
In particular, all $H_\mu$ with $\mu \in \slz \backslash \SL{2}$ have distance strictly larger than $m_r^1$, $r=1,2$.
\end{Proposition}
\begin{proof} 
For brevity let us call any deformation $H_\mu$ and the corresponding change of basis $\mu$ that has deformation distance $m^1_r, \, r=1,2$ an \emph{excited state}. The proof follows along the same lines as the proof of Theorem~\ref{ThBain}. First we show with the help of a computer that the second to minimal deformation distance within $\SL{2}$ is given by the above and by \eqref{d1FinalLB} and \eqref{d2FinalLB} respectively applied on $\slz \backslash \SL{2}$ we know that there cannot be any excited states in $\slz\backslash \SL{2}$.
\end{proof}

\begin{Corollary}{(``The first excited state'' with volume change)}\label{CorEx}\\
 In a transformation from an fcc-to-bcc lattice with volume change $\la^3$ the second to minimal deformation distances are given by
\begin{alignat*}{2}
m_1^{1,\la}&\colonequals \min_{\mu\in \slz \atop \mu \neq \mu_{\min}}\dm_1 ({H^\la_\mu},\id)= 2^{-3/2}{\sqrt{25\  {2}^{1/3} \la^2-4\ 2^{2/3} \left(4+\sqrt{17}\right) \la+24}},\\ 
 m_2^{1,\la}&\colonequals \min_{\mu\in \slz \atop \mu \neq \mu_{\min}}\dm_2 ({H^\la_\mu},\id)= 2^{-3}\ \sqrt {305\ {2}^{2/3}\la^4-400\ {2}^{1/3}\la^2+192}.
\end{alignat*}
In particular all $H^\la_\mu$ with $\mu \in \slz \backslash \SL{2}$ have distance strictly larger than $m_r^{1,\la}$, $r=1,2$.
\end{Corollary}
\begin{Theorem}\label{ThStab}
 The Bain strain is a $\dm_1$- and $\dm_2$-optimal lattice transformation from fcc-to-bct with lattice parameters $A, C$ in the
 range 
 \begin{alignat*}{3}
   &\lbrace (A,\,C): C\geq A>0.75 \aand m^1_{1}-3^{3/2}|B_\ac-B|\geq m_{1}^{0,\ac} \rbrace &&\quad\for r=1,\\
   &\lbrace (A,\,C): C\geq A>0.75 \aand m^1_{2}-27|B_\ac\T B_\ac-B\T B|\geq m_{2}^{0,\ac}  \rbrace &&\quad\for r=2,
 \end{alignat*}
 (cf. Figure~\ref{Figd0}). For $C>A$ the stretch components of the optimal lattice transformations are given by
$\bar H^\ac_{\min}$ in the set
\begin{alignat}{3}\no
\small{\left\lbrace \begin{pmatrix} 2^{1/6}A & 0 & 0 \\ 0 & 2^{1/6}A & 0\\ 0 & 0 & 2^{-1/3}C \end{pmatrix},\begin{pmatrix} 2^{1/6}A & 0 & 0 \\ 0 & 2^{-1/3}C & 0\\ 0 & 0 & 2^{1/6}A \end{pmatrix},
\begin{pmatrix} 2^{-1/3}C & 0 & 0 \\ 0 & 2^{1/6}A & 0\\ 0 & 0 & 2^{1/6}A \end{pmatrix}\right\rbrace.}
\end{alignat}
The respective minimal metric distances are  
\begin{alignat}{3}
 m^\ac_{\min,1}&=\dm_1({ H^\ac_{\min}},\id)=\left(2\, \left( 2^{1/6}A-1 \right) ^{2}+ \left( 2^{-1/3}C-1 \right) ^{2}\right)^{1/2}=m_{1}^{0,\ac}, \label{Eqmacmin1}\\
 m^\ac_{\min,2}&=\dm_2({ H^\ac_{\min}},\id)=\left(2\, \left( \left(2^{1/6}A\right)^{2}-1 \right) ^{2}+ \left( \left(2^{-1/3}C\right)^2-1 \right) ^{2}\right)^{1/2}=m_{2}^{0,\ac} \label{Eqmacmin2}
\end{alignat}
and are achieved by exactly $24$ distinct $\mu \in \slz$. The case $C=A$ corresponds to an fcc-to-bcc transformation with volume change $\la^3$ with $\la=A=C$ and we refer to Corollary~\ref{Coralpha}. 
\end{Theorem}
\begin{figure}[h]
\centering
\begin{minipage}{.5\textwidth}
  \centering
  \vspace{.05cm}
  \includegraphics[width=.9\linewidth]{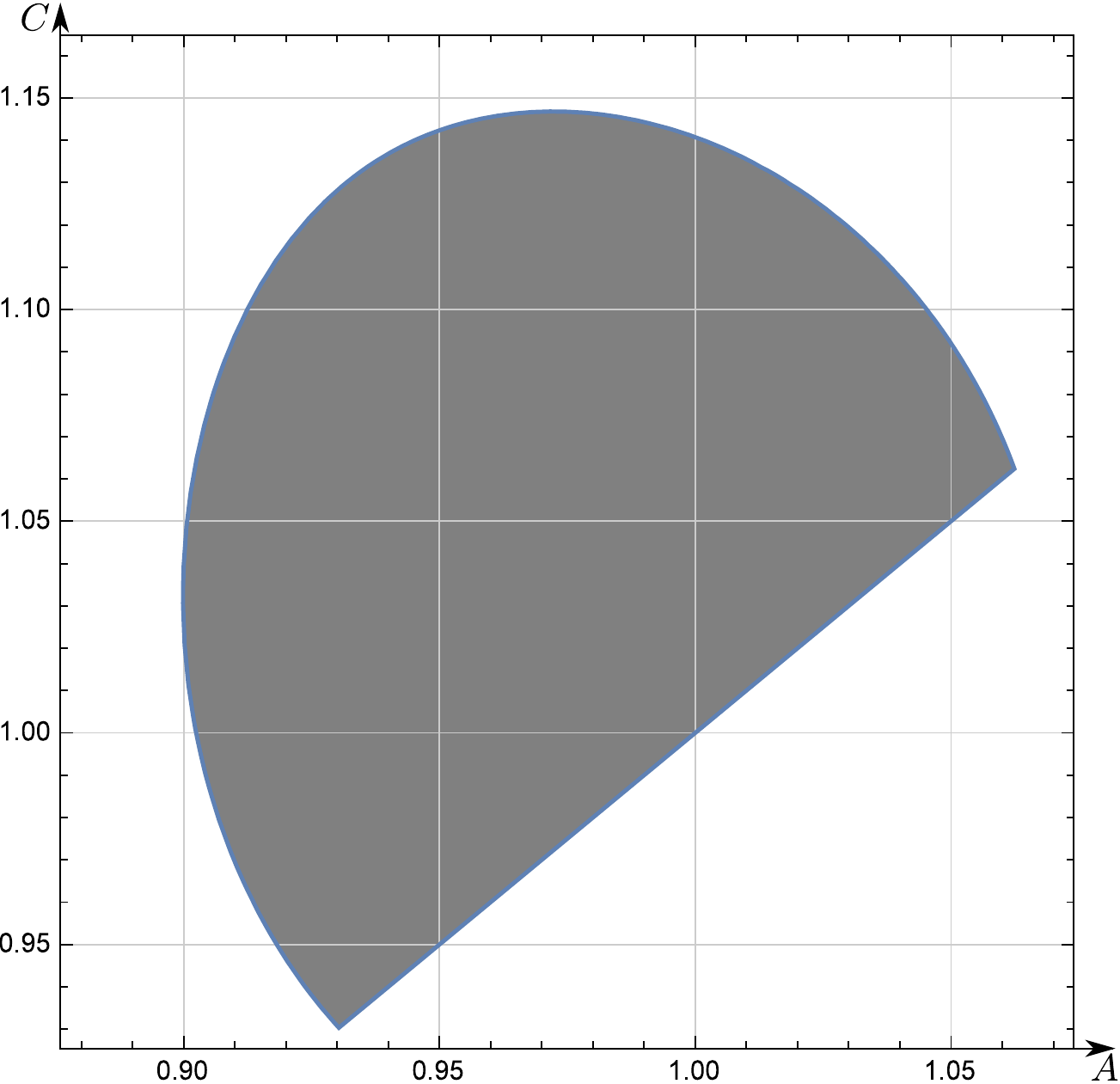}
\end{minipage}%
\begin{minipage}{.5\textwidth}
  \centering
  \includegraphics[width=.9\linewidth]{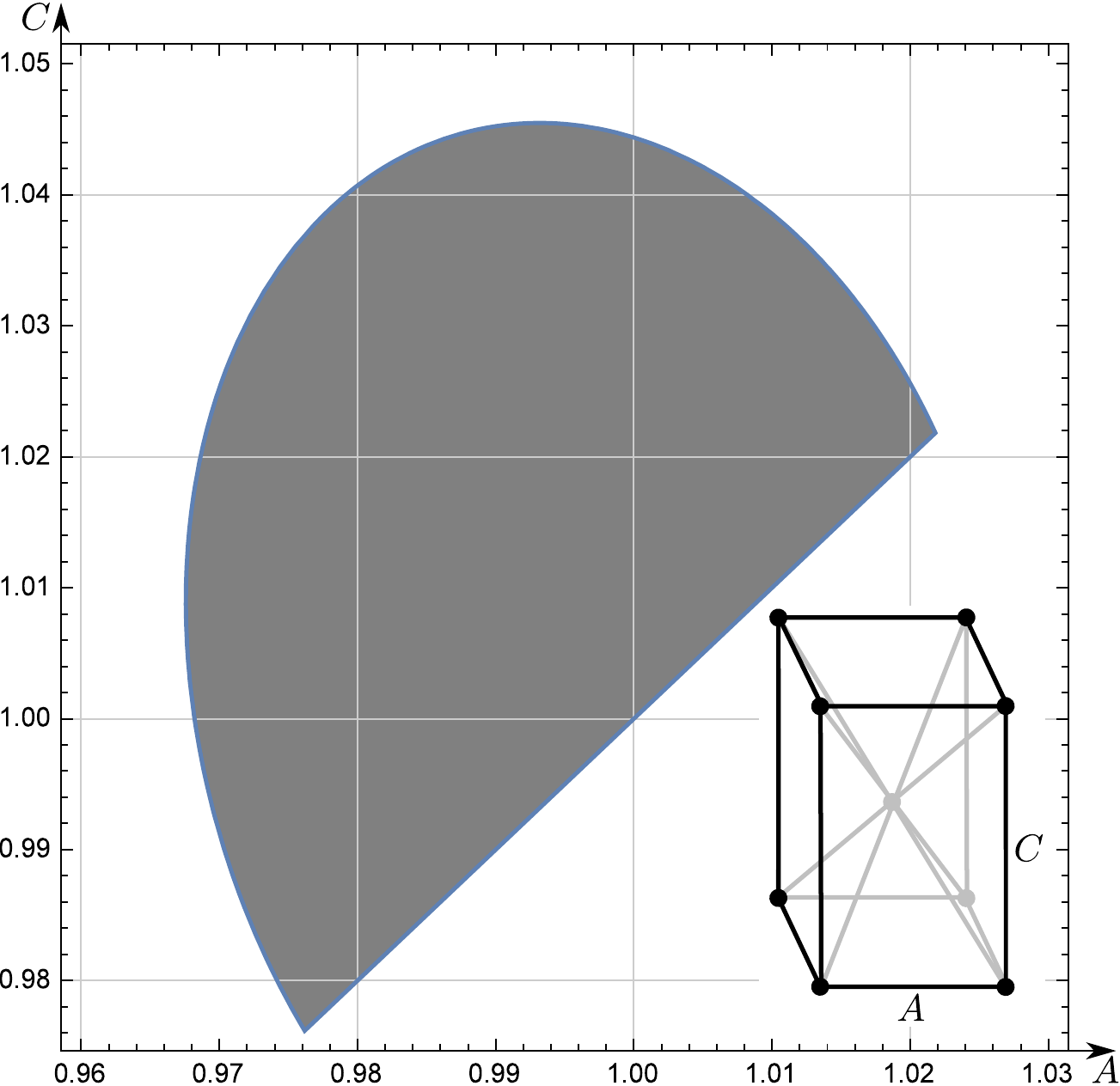}
\end{minipage}
\caption{The range of $A,C$ where the Bain strain remains $\dm_1$-optimal (left) and $\dm_2$-optimal (right).}\label{Figd0}
\end{figure}
\begin{Example}
   $A=C=1$ recovers Theorem~\ref{ThBain}.
   $C>A$ corresponds to the usual fcc-to-bct transformation for steels with higher carbon content.
   $C=\sqrt{2}A=2^{1/3}$ is the bct lattice that is contained in the fcc lattice, i.e. $\dm(\ml_0,\ml_1)=0$.
\end{Example}
\begin{proof}[Proof (of Theorem~\ref{ThStab})]
We will show that precisely 24 of the 72 $\mu$'s that were optimal in the fcc-to-bcc transition remain optimal. Let us start from one of those optimal transformations from fcc-to-bcc given by e.g.
 \begin{equation}\no
  \mu_0=\begin{pmatrix} 1 & 1 & 1 \\ 0 & 1 & 0\\ 0 & 1 & 1 \end{pmatrix}.
 \end{equation}
We know that $H^\ac_{\mu_0}=B_\ac\mu_0F^{-1}$ is optimal for $A=C=1$, where $B_\ac$ is given by \eqref{Optbct}. The deformation distance is $\dm_r({ H^\ac_{\mu_0}},\id)=m_{\min,r}^\ac$ with $m_{\min,r}^\ac$ given by \eqref{Eqmacmin1} and \eqref{Eqmacmin2} respectively. Further there exist $24$ different $\mu$'s and corresponding $H^\ac_\mu$'s that have the same distance. This follows as in Remark~\ref{PropRel} with the only difference that the point group of the bct lattice only has $8$ elements. Then, due to the invariance of $\dm_1$ and $\dm_2$ under multiplication from the left or right by any rotation, the 24 matrices $H^\ac_\mu$ trivially have the same distance to $\id$ and one may 
easily verify that these $H^\ac_\mu$'s  are equally split into the three equivalence classes as in the statement of the Theorem.

The remaining $48$ $\mu$'s that were optimal for fcc-to-bcc lead to a larger deformation distance. One of these non-optimal lattices is generated by
  \begin{equation}\no
  \tilde B_\ac=2^{-4/3}   \left( \begin {array}{ccc} -A&A&0\\ \noalign{\medskip}
A&A&0\\ \noalign{\medskip}-C&-C&-2C
\end {array} \right)
 \end{equation}
with
\begin{alignat}{3}\no
 &\left(\dm_1({\tilde B_{\ac} F^{-1}},\id)\right)^2-\left({m^\ac_{\min,1}}\right)^2={2}^{-5/3} \left( C-A \right)  \left( 2^{1/6}\,{A}^
{2}+2^{1/6}\,{C}^{2}-4+2^{3/2} \right)>0 \\
&\qquad \Leftarrow \ C>A \aand A^2+C^2 > \frac{2^{5/6}}{2-\sqrt{2}} \quad \Leftarrow \ C>A>0.75 \label{EqLowerAbound1}
\end{alignat}
and
\begin{alignat}{3}\no
 &\left(\dm_2({\tilde B_{\ac} F^{-1}},\id)\right)^2-\left({m^\ac_{\min,2}}\right)^2={2}^{-4/3} \left( C-A \right)  \left( A+C \right)  \left( 3\,{A}^
{2}+3\,{C}^{2}-2\,{2}^{2/3} \right)>0 \\
&\qquad \Leftarrow \ C>A \aand A^2+C^2 > \frac{2^{8/3}}{3} \quad \Leftarrow \ C>A>0.75 \label{EqLowerAbound2}
\end{alignat}
which holds true for all $C>A$ in the range under consideration. The remaining $47$ non-optimal deformations ${H_{\mu}^{\ac}}$ are all $\PP$ related and thus have the same distance; in particular larger than $m_{\min,r}^\ac$.

To show the minimality of the 24 $H^\ac_\mu$'s, we make use of our result on the first excited state to compare their distance to $\id$ against all the remaining $\mu$'s that were non-optimal in the fcc-to-bcc transition. In particular, we need to show that 
\begin{equation}\no
 m^\ac_{\min,i}<\min_{\mu \neq \mu_{\min}}\dm_r({{H_{\mu}^{\ac}}},\id), \quad r=1,2,
\end{equation}
where ${H_{\mu}^{\ac}}=B_\ac \mu F^{-1}$ and $\mu_{\min}$ is any of the $72$ minimising $\mu$'s from Theorem~\ref{ThBain}. Let us set $H_\mu\colonequals H_{\mu}^{\mbox{\tiny 11}}$ and estimate using the properties of $\dm_r$ (cf. Example~\ref{ExDi})
\begin{alignat}{3} \label{dtriangleest}
\min_{\mu \neq \mu_{\min}}\dm_r({{H_{\mu}^{\ac}}},\id)&\geq \min_{\mu \neq \mu_{\min}}\left(\dm_r({H_{\mu}},\id)-\dm_r({{H_{\mu}^{\ac}}},H_\mu)\right)\\ \no
 &\geq m_{r}^1-\max_{\mu \neq \mu_{\min}}\dm_r({{H_{\mu}^{\ac}}},H_\mu),
\end{alignat}
where $m_{r}^1$ denotes the first excited state (cf. Proposition~\ref{PropEx}). We estimate 
\begin{alignat*}{2}
 \dm_1({{H_{\mu}^{\ac}}},H_\mu) &\leq  |{H_{\mu}^{\ac}}-H_{\mu}| \leq |B_{\ac}-B||\mu F^{-1}|,\\
 \dm_2({{H_{\mu}^{\ac}}},H_\mu) &\leq |{H_{\mu}^{\ac}}\T {H_{\mu}^{\ac}}-H_{\mu}\T H_{\mu}| \leq |B\T _{\ac}B_{\ac}-B\T B||\mu F^{-1}|^2
\end{alignat*}
and with the help of a computer we calculate $\max_{\mu \in \SL{1}}|\mu F^{-1}|=3^{3/2}$. Thus a sufficient condition that the Bain strain is $\dm_1$-optimal within $\SL{1}$ is that $0.75<A\leq C$ satisfy
\begin{equation} \no
 m^1_{1}-3^{3/2}|B_\ac-B|\geq m_{\min,1}^\ac
\end{equation}
yielding the area drawn in Figure~\ref{Figd0} (left). To exclude any $\mu \in \slz \backslash \SL{1}$ we replace $b$ by $b_\ac=\diag(A,A,C)b$ in \eqref{Eqb} in the proof of Theorem~\ref{ThBain} and estimate $|b_\ac|\geq A |b|$. Concluding as in \eqref{d1FinalLB} we arrive at 
\begin{equation} \label{CondOptSLk}
 \min_{\mu \in \slz \backslash \SL{1}}\dm_1({H_{\mu}^{\ac}},\id) \geq  2^{2/3}A-1, 
\end{equation}
which needs to be larger than ${m_{\min,1}^{\ac}}$. This holds true e.g. for $A\in [0.85,1.7]$ and $C\in[A,1.7]$. The proof of the $\dm_2$-optimality proceeds analogously. To obtain $\dm_2$-optimality within $\SL{1}$ we need to satisfy
\begin{equation} \no
 m^1_{2}-27|B_\ac\T B_\ac-B\T B|\geq m_{\min,2}^\ac
\end{equation}
for all $0.75<A\leq C$ which yields the area drawn in Figure~\ref{Figd0} (right). To exclude all elements in the complement of $\SL{1}$ we have to ensure that $2^{4/3}A-1{>}{m_{\min,2}^{\ac}}$ which holds true e.g. for $A\in [0.75,1.5]$ and $C\in[A,1.5]$.
\end{proof}
\begin{figure}[h]
\centering
\begin{minipage}{.5\textwidth}
  \centering
  \includegraphics[width=.9\linewidth]{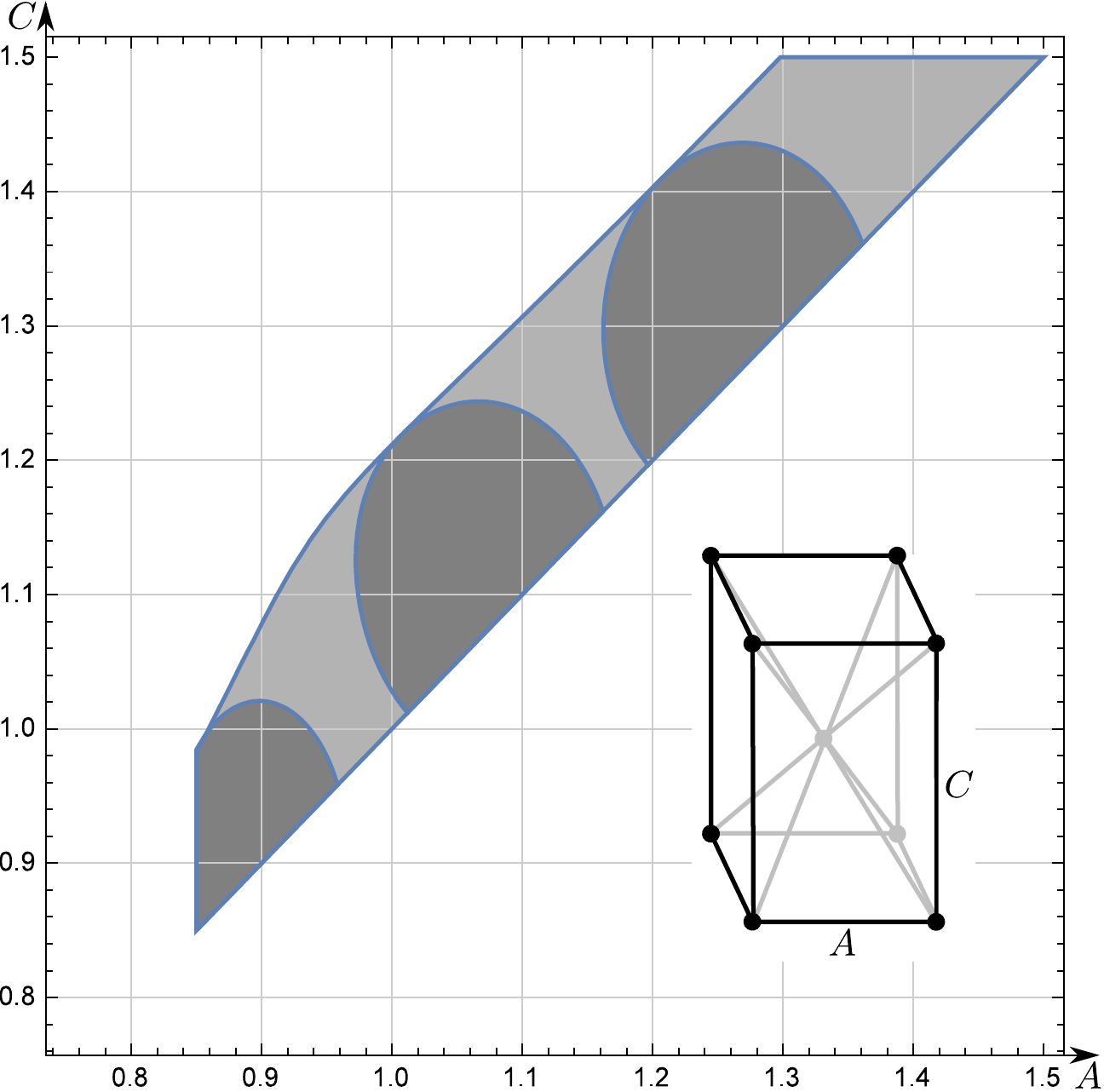}
  \label{Figd1}
\end{minipage}%
\begin{minipage}{.5\textwidth}
  \centering
  \includegraphics[width=.9\linewidth]{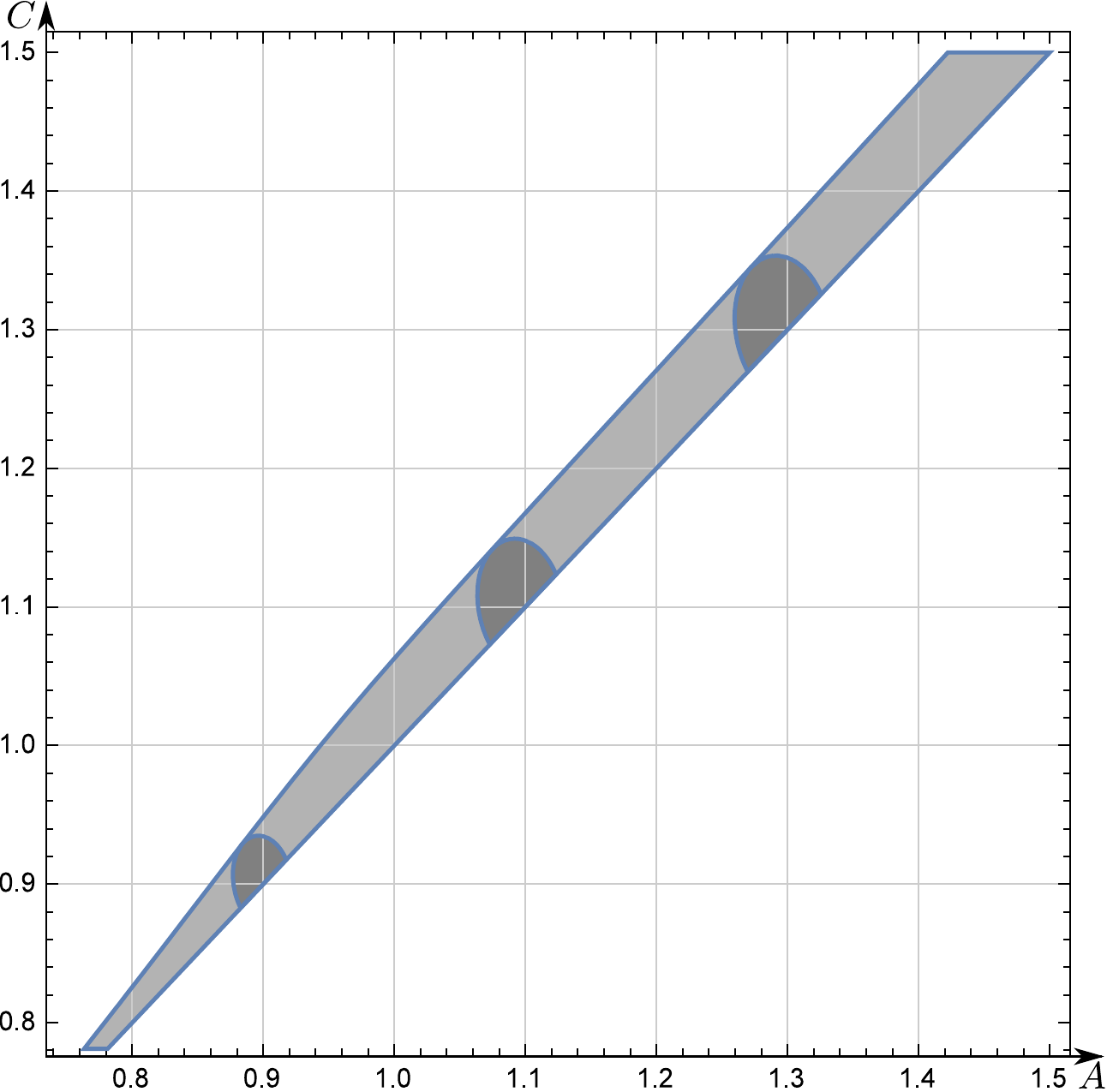}
  \label{Figd2}
\end{minipage}
\caption{Extended $\dm_1$-optimality (left) and $\dm_2$-optimality (right) range for $A$ and $C$. The dark shaded regions are obtained by a perturbation argument about the fixed optimal transformations $H^{\la_0}_{\min}$, $\la_0 = 0.9, 1.1, 1.3$. The light shaded regions are obtained by a perturbation argument about the $(A,C)$-dependent optimal transformation $H_{\min}^{\la(A,C)}$ with $\la(A,C)=0.995\sqrt{AC}$.}\label{FigExt}
\end{figure}
\begin{Remark}
By employing Corollary~\ref{Coralpha} and Corollary~\ref{CorEx} we may extend the range of optimal parameters $A$ and $C$ in Theorem~\ref{ThStab} (cf. Figure~\ref{Figd0}) by shifting the reference point from $A=C=1$ to $A=C=\la$. This enables us to show that the Bain strain remains the $\dm_1$- and $\dm_2$-optimal lattice deformation from fcc-to-bct in a much larger range of lattice parameters $C\geq A$. The shaded regions in Figure~\ref{FigExt} show the values of $A,C$ such that the deformation $H^{\ac}_{\min}$ remains $\dm_1$- and $\dm_2$-optimal. 
\end{Remark}
\begin{proof}
The main idea of the proof is to consider the element $H_\mu^\la=\la H_\mu$ which is optimal in an fcc-to-bcc transformation with volume change (cf. Corollary~\ref{Coralpha}) and ``closest'' to $H^\ac_{\mu}$. Thus in \eqref{dtriangleest} we write
\begin{alignat*}{3} 
\min_{\mu \neq \mu_{\min}}\dm_r({{H_{\mu}^{\ac}}},\id)&\geq \min_{\mu \neq \mu_{\min}}\left(\dm_r({H^\la_{\mu}},\id)-\dm_r(H^\la_\mu,{H^\ac_{\mu}})\right)
\end{alignat*}
and for $\dm_1$ we estimate
\begin{equation} \no
  \min_{\mu \neq \mu_{\min}}\dm_1({H^\ac_{\mu}},\id)\geq m_1^{1,\la} - \max_{\mu \neq \mu_{\min}} \left| \la H_{\mu}-H^\ac_{\mu}\right|\geq m_1^{1,\la} - \la \max_{\mu \neq \mu_{\min}} |\mu F^{-1}|| B-B_{\frac{A}{\la} \frac{C}{\la}} |.\no 
\end{equation}
To obtain the $\dm_1$-optimality in $\SL{1}$, we again use $\max_{\mu \in \SL{1}}|\mu F^{-1}|=3^{3/2}$ so that $m_1^{1,\la}-\la 3^{3/2}|B-B_{\frac{A}{\la} \frac{C}{\la}}|\geq m_{\min,1}^\ac$. Condition \eqref{CondOptSLk} regarding $\dm_1$-optimality outside $\SL{1}$ remains unchanged. Both conditions combined yield the area shown in Figure~\ref{FigExt} (left). Analogously in order to show $\dm_2$-optimality we estimate 
\begin{alignat*}{2}
 \min_{\mu \neq \mu_{\min}}\dm_2({H^\ac_{\mu}},\id)&\geq m_2^{1,\la} - \max_{\mu \neq \mu_{\min}} \left| \la^2 H_{\mu}\T H_{\mu}-{H^\ac_{\mu}}\T H^\ac_{\mu}\right| \\
 &\geq m_2^{1,\la} - \la^2 \max_{\mu \neq \mu_{\min}} |\mu F^{-1}|^2 \left| B\T  B-B_{\frac{A}{\la} \frac{C}{\la}}\T B_{\frac{A}{\la} \frac{C}{\la}}\right|
\end{alignat*}
and again use $\max_{\mu \in \SL{1}}|\mu F^{-1}|=3^{3/2}$ to show $\dm_2$-optimality within $\SL{1}$. The condition to be $\dm_2$-optimal outside $\SL{1}$ remains unchanged. Both conditions combined yield the area shown in Figure~\ref{FigExt} (right).
\end{proof}
\begin{Remark}
The previous estimates can be iterated, i.e. instead of picking the $H^\la_\mu$ that is closest to $H_{\mu}^{\ac}$ one may pick any $H_{\mu}^{\mbox{\tiny{A'C'}}}$ that is in the range indicated in Figure~\ref{FigExt} that is closest to $H_{\mu}^{\ac}$.

If $C\leq A$, following the proof of Theorem~\ref{ThStab}, one finds that the optimal strain becomes
 \begin{equation}\no
   \bar H_{\min}= \begin{pmatrix}
                   2^{1/6}\frac{A+C}{2} & \pm 2^{1/6}\frac{A-C}{2} &0\\
                   \pm 2^{1/6}\frac{A-C}{2} & 2^{1/6}\frac{A+C}{2} &0\\
                   0 & 0& 2^{-1/3}A
                  \end{pmatrix} + \mbox{ its $\PP$ conjugates}
 \end{equation}
 at least if $A$ and $C$ are in the regions specified in the statement of Theorem~\ref{ThStab} with $C\geq A$ replaced by $A\geq C$.
\end{Remark}
\subsection{Terephthalic Acid}\label{subsecTere}
Terephthalic Acid is a material that has two triclinic phases (Type I and II) which are very different from each other (cf. \cite[p.46 ff.]{BJInco}). Thus any lattice transformation necessarily requires large principal stretches and, unlike in the Bain setting, it is not clear what a good candidate for the optimal transformation would be. However, with the help of the proposed framework the $\dm_r$-optimal lattice transformation can easily be determined. The only required input parameters are the lattice parameters of the two triclinic unit cells (cf. \cite[p.388 Table~2]{TereAcid}) listed in Table~\ref{TableTereAcid}.
\begin{table}[h]
\begin{center}
  \begin{tabular}{l*{6}{c}}
Form & $a/A^\circ$ &$b/A^\circ$ &$c/A^\circ$ &$\al/^\circ$ & $\beta/^\circ$ & $\gamma/^\circ$ \\
\hline
I & 7.730 & 6.443&3.749&92.75 & 109.15 & 95.95\\
II & 7.452&6.856 & 5.020 & 116.6 &119.2 & 96.5
\end{tabular}
\end{center}
\caption{Lattice parameters of the triclinic unit cells of Terephthalic Acid}\label{TableTereAcid}
\end{table}

To apply our analysis we first ought to convert the triclinic to the primitive description. 
\begin{Lemma}{(Conversion from triclinic to primitive unit cell)} \label{LemmaConv} \\
The triclinic unit cell with lattice parameters $a,b,c$ and $\al,\beta, \gamma$ generates up to an overall rotation the same lattice as 
\begin{equation} \no
 F=\left(
\begin{array}{ccc}
 a & b \cos \left(\gamma   \right) & c \cos \left(\beta   \right)
   \\
 0 & b \sin \left(\gamma   \right) & c\sin^{-1} \gamma  \left(\cos \left(\alpha  
   \right)-\cos \left(\beta   \right) \cos \left(\gamma  
   \right)\right)  \\
 0 & 0 & c \left(\sin^2 \beta-\sin^{-2} \gamma\left(\cos
   \left(\alpha   \right)-\cos \left(\beta   \right) \cos
   \left(\gamma   \right)\right)^2\right)^{1/2} \\
\end{array}
\right)
\end{equation}
\end{Lemma}

\begin{proof}
Let $\col F = \{f_1,f_2,f_3\}$. It is easy to verify that $|f_1|=a$, $|f_2|=b$, $|f_3|=c$ and that $\measuredangle (f_1,f_2) = \gamma$, $\measuredangle (f_1,f_3) = \beta$, $\measuredangle (f_2,f_3) = \alpha$.
\end{proof}

\begin{Example}{(Primitive cells of Terephthalic Acid)}\\
Application of the previous Lemma to the lattice parameters in Table~\ref{TableTereAcid} leads to
\begin{alignat}{3}\label{TereFG}
 F_I=\left(
\begin{array}{ccc}
 7.730 & -0.668 & -1.230 \\
 0 & 6.408 & -0.309 \\
 0 & 0 & 3.528 \\
\end{array}
\right) \aand F_{II}=\left(
\begin{array}{ccc}
 7.452 & -0.776 & -2.449 \\
 0 & 6.812 & -2.541 \\
 0 & 0 & 3.570 \\
\end{array}
\right)
\end{alignat}
(all measures in $A^\circ$) 
\end{Example}
\begin{Theorem}{(Optimal lattice transformations in Terephthalic Acid)}\\
 The unique equivalence class of $d_1$- and $d_2$-optimal transformations between Terephthalic Acid Form I and Terephthalic Acid Form II has a stretch component given by
 \begin{equation}\label{HminAcid}
  \bar H_{\min}=\left(
\begin{array}{ccc}
 0.820 & -0.125 & -0.072 \\
 -0.125 & 0.994 & -0.146 \\
 -0.072 & -0.146 & 1.329 \\
\end{array}
\right) \end{equation}
 with principal stretches $\nu_1=0.725$, $\nu_2=1.033$ and $\nu_3=1.385$. The minimal distances are given by
 $m_{\min,1}=d_1(\bar H_{\min},\id)=0.474$ and $m_{\min,2}=d_2(\bar H_{\min},\id)=1.035$.
\end{Theorem}
\begin{proof}
We apply Theorem~\ref{TheCpct} with $F=F_I$ and $G=F_{II}$. We calculate $\tinf{F_I}=|F_Ie_1|=7.730<7.8$, $\nu_{\min}(F_{II})=3.076>3$, $m_{0,1}=0.529<0.55$ and $m_{0,2}=1.197<1.2$. Further we note that if $\mu \in \slz \backslash \SL{k-1}$ then $\tinf{\mu}\geq \sqrt{k^2+1}$. Therefore by \eqref{Eqmu2} any $\dm_1$-optimal $\mu$ satisfies
\begin{alignat*}{2}
 \tinf{\mu}\leq \frac{\tinf{F}}{\nu_{\min}(G)}(m_{0,1}+1) <\frac{7.8}{3}(0.55+1)=4.03<\sqrt{4^2+1}
\end{alignat*}
and is thus contained in $\SL{3}$ and by \eqref{Eqmu2} any $\dm_2$-optimal $\mu$ satisfies
\begin{alignat*}{2}
 \tinf{\mu}^2 \leq \frac{\tinf{F}^2}{\nu_{\min}^2(G)}(m_{0,2}+1) <\frac{7.8^2}{3^2}(1.2+1)=14.872
\end{alignat*}
and is thus also contained in $\SL{3}$. With the help of a computer we check that within the set $\SL{3}$ the minimum is in both cases attained at 
\begin{equation} \no
 \mu_{\min}=\left(
\begin{array}{ccc}
 0 & 1 & 0 \\
 1 & 0 & 0 \\
 1 & 1 & -1 \\
\end{array}
\right)
\end{equation}
with corresponding $\bar H_{\min}=\sqrt{H\T _{\mu_{\min}}H_{\mu_{\min}}}$ as in \eqref{HminAcid}, $m_{\min,1}=\dm_1(\bar H_{\min},\id)=0.474$ and $m_{\min,2}=\dm_2(\bar H_{\min},\id)=1.035$.
\end{proof}
\begin{Remark}
 We have shown that the $d_1$- and $d_2$-optimal transformations from Form I to Form II of Terephthalic Acid are the same. However the $d_{-2}$-optimal transformation is different and given by
 \begin{equation} \no
  \bar H_{\min,\, -2}=\left(
\begin{array}{ccc}
 0.852 & -0.119 & -0.018 \\
 -0.119 & 0.950 & -0.197 \\
 -0.018 & -0.197 & 1.346 \\
\end{array}
\right) \end{equation}
with principal stretches $\nu_1=0.743, \, \nu_2=0.977$ and $\nu_3=1.429$. As expected the smallest principal stretch $\nu_1$ is bigger than before since $d_{-2}$ penalises contractions significantly more than expansions. To obtain the required analytical bounds one calculates $\tinf{F_{II}}=|F_{II}e_1|=7.452<7.46$, $\nu_{\min}(F_{I})= 3.464>3.45$ and $m_{0,-2}=1.080<1.1$ to get $\tinf{\mu^{-1}}^2<10$ (cf. \eqref{Eqmu-2}) and thus the optimal $\mu$ lies in $\SL{-2}$. 
\end{Remark}
\begin{Remark}
Even though \cite{Sherry} also uses the distance $d_{-2}$, the reported strains are different. However, this discrepancy is not a consequence of using sublattices but rather a consequence of using the lattice parameters from \cite[p.388 Table~1]{TereAcid} for Terephthalic Acid II\footnote{$a=9.54A^\circ$, $b=5.34A^\circ$, $c=5.02A^\circ$,
$\alpha= 86.95^\circ, \beta= 134.65^\circ, \gamma = 94.8^\circ$} and the lattice parameters from \cite[p.388 Table~2]{TereAcid} for Terephthalic Acid I. In particular, using these mixed input parameters in ``OptLat'' (cf. \hyperref[Matlab]{Appendix}) yields identical values.
\end{Remark}

\section{Concluding remarks}
The present paper provides a rigorous proof for the existence of an optimal lattice transformation between any two given Bravais lattices with respect to a large number of optimality criteria. Furthermore, a precise \hyperref[Mathematica]{algorithm} and a \hyperref[Matlab]{GUI} to determine this optimal transformation is provided (cf. \hyperref[append]{Appendix}). As possible applications, the optimal transformation in steels, that is the transformation from fcc-to-bcc/bct, and in Terephthalic Acid were determined. Through Theorem~\ref{TheCpct} and with the help of the provided algorithm/programme one is able to rigorously determine the optimal phase transformation in any material undergoing a displacive phase transformation from one Bravais lattice to another. 

If the parent or product phases are multilattices, the proposed framework is not \emph{a priori} applicable. Nevertheless, one may still consider Bra\-vais sublattices of these multilattices and proceed as before. The choice of these sublattices may come from physical consideration. However, in order to rigorously determine the optimal transformation between two given multilattices one would need to measure the movement of \emph{all} atoms consistently, that is one would need to take into account both the overall periodic deformation of the unit cell and the shuffle movement of atoms within the unit cell. Establishing such a criterion would be of great interest but lies outside the scope of the present paper.

\appendix
\section*{Appendix}\label{append}
\subsection*{A. Mathematica}\label{Mathematica}
The following Mathematica code\footnote{The original .nb file can be found online under \href{http://solids.maths.ox.ac.uk/programs/OptLat.nb}{http://solids.maths.ox.ac.uk/programs/OptLat.nb}} determines for a given $k\in \mathbb{N}$ the optimal lattice transformation $H:\ml(F) \ra \ml(G)$ within the set $\{H_\mu=G\mu F^{-1}:\, \mu\in \SL{k}\}$ for any given $F,G \in \glp$ and for any distance measure $d_r(H,\id)$, $r>0$. The case $r<0$ is analogous.

For the transformation from fcc-to-bcc, $F$ and $G$ would be given by \eqref{FB} and for the transformation from Terephthalic Acid I to II, $F$ and $G$ would be given by \eqref{TereFG}.

Firstly, we generate the set $\SL{k}${\fontfamily{qcr}\selectfont  (=SL)}:

\noindent{\fontfamily{qcr}\selectfont \small
SL = Select[Flatten[Table[\{a,b,c,d,e,f,g,h,i\},\\ \hspace*{1.2cm}\{a,-k,k\},\{b,-k,k\},\{c,-k,k\},\{d,-k,k\},\{e,-k,k\},\{f,-k,k\},\\ \hspace*{1.2cm}\{g,-k,k\},\{h,-k,k\},\{i,-k,k\}],8],Det[Partition[\#,3]]==1\&];
}

Next we generate a list{\fontfamily{qcr}\selectfont  (=distlist)} of all values of $\dm_r(H_\mu,\id)$ for $\mu\in \SL{k}$:\\
\noindent{\fontfamily{qcr}\selectfont \small
 Hmu = Function[mu,G.Partition[mu,3].Inverse[F]];\\
distr = Function[mu,Norm[SingularValueList[Hmu[mu]]\string^r-\{1,1,1\}]];\\
distlist = distr/@SL;}

Then we generate a list{\fontfamily{qcr}\selectfont  (=poslist)} of all the positions of $\mu$'s in $\SL{k}$ that give rise to the minimal deformation distance: \\
\noindent{\fontfamily{qcr}\selectfont \small
poslist = Flatten[Position[distlist,RankedMin[distlist,1]],1];}

Further we calculate the minimal deformation distance $m_0$, the second to minimal deformation distance $m_1$ and return their numerical difference $\Delta=m_1-m_0${\fontfamily{qcr}\selectfont  (=delta)}:\\
\noindent{\fontfamily{qcr}\selectfont  \small
m0 = distlist[[poslist[[1]]]]; m1 = Sort[distlist][[Length[poslist]+1]];\\
delta = N[m1-m0]}

Finally we return a list of all $\mu$'s that give rise to an optimal deformation $H_\mu$ and a list of all optimal $H_\mu$'s:\\
{\fontfamily{qcr}\selectfont  \small 
SL[[poslist]] \\
Hmu/@SL[[poslist]]}
\subsection*{B. MATLAB}\label{Matlab}
A Graphical User Interface (GUI) called ``OptLat'' can either be found on \mbox{MATLAB} \href{http://uk.mathworks.com/matlabcentral/fileexchange/55554-optlat}{File Exchange}\footnote{\href{http://uk.mathworks.com/matlabcentral/fileexchange/55554-optlat}{http://uk.mathworks.com/matlabcentral/fileexchange/55554-optlat}} (requires MATLAB) or downloaded directly as a standalone \href{http://solids.maths.ox.ac.uk/programs/OptLat.exe}{\mbox{Windows} application}\footnote{\href{http://solids.maths.ox.ac.uk/programs/OptLat.exe}{http://solids.maths.ox.ac.uk/programs/OptLat.exe}}. 
\vspace*{2em}

{\noindent {\bf Acknowledgements:}
We would like to thank R. D. James and X. Chen for helpful discussions. The research of A. M. leading to these results has received funding from the European Research Council under the European Union's Seventh Framework Programme (FP7/2007-2013) / ERC grant agreement n$^\circ\, 291053$.}

\end{document}